\documentclass[11pt, onecolumn]{IEEEtran}

\usepackage{graphicx}

\usepackage{amsmath}
\usepackage{amssymb}
\usepackage{amsfonts}
\usepackage{epsfig}

\usepackage{amsmath}    
\usepackage{graphicx}   
\usepackage{color}      
\usepackage{hyperref}
\usepackage{comment}
\usepackage{subfigure}
\usepackage{times}
\usepackage{hyphenat}
\usepackage{epsfig}
\usepackage{latexsym}
\usepackage{bbold,bbm}
\usepackage{setspace}
\usepackage{amssymb,amsmath}
\usepackage{mathrsfs}
\usepackage{amsgen,amsfonts,amsbsy,amsthm}
\usepackage{algorithmic,algorithm}
\usepackage{array}
\usepackage{booktabs}
\usepackage{amsfonts}
\usepackage{amssymb}
\usepackage{multirow}
\usepackage{longtable}
\usepackage{epstopdf}
\usepackage{wrapfig}
\usepackage[font=footnotesize]{caption}

\newtheorem{example}{Example}
\newtheorem{defn}{Definition}
\newtheorem{thm}{Theorem}

\newtheorem{cor}[thm]{Corollary}
\newtheorem{note}{Remark}

\newtheorem{conj}{Conjecture}
\newtheorem{const}{Construction}
\newtheorem{claim}{Claim}

\newcommand{\bit}{\begin{itemize}}
\newcommand{\eit}{\end{itemize}}
\newcommand{\bcor}{\begin{cor}}
\newcommand{\ecor}{\end{cor}}
\newcommand{\beq}{\begin{equation}}
\newcommand{\eeq}{\end{equation}}
\newcommand{\beqn}{\begin{equation*}}
\newcommand{\eeqn}{\end{equation*}}
\newcommand{\bea}{\begin{eqnarray}}
\newcommand{\eea}{\end{eqnarray}}
\newcommand{\bean}{\begin{eqnarray*}}
\newcommand{\eean}{\end{eqnarray*}}
\newcommand{\ben}{\begin{enumerate}}
\newcommand{\een}{\end{enumerate}}
\newcommand{\bdefn}{\begin{defn}}

\renewcommand\footnotemark{}

\newcommand{\codec}{\ensuremath{{\cal C}}}

\newcommand\scalemath[2]{\scalebox{#1}{\mbox{\ensuremath{\displaystyle #2}}}}

\begin{document}
\sloppy
\title{Binary Codes with Locality for Multiple Erasures Having Short Block Length}

\author{
  \IEEEauthorblockN{S. B. Balaji, K. P. Prasanth and P. Vijay Kumar, \it{Fellow}, \it{IEEE}}
  
  \IEEEauthorblockA{Department of Electrical Communication Engineering, Indian Institute of Science, Bangalore.  \\ Email: balaji.profess@gmail.com, prasanthkp231@gmail.com,pvk1729@gmail.com} 
\thanks{P. Vijay Kumar is also an Adjunct Research Professor at the University of Southern California.  This work is supported in part by the National Science Foundation under Grant No. 1421848 and in part by the joint UGC-ISF research program.}
}
\maketitle

\begin{abstract}
	This paper considers linear, binary codes having locality parameter $r$, that are capable of recovering from $t\geq 2$ erasures and which additionally, possess short block length.  Both parallel (through orthogonal parity checks) and sequential recovery are considered here. In the case of parallel repair, minimum-block-length constructions are characterized whenever $t |(r^2+r)$ and examples examined.   In the case of sequential repair, the results include (a) extending and characterizing minimum-block-length constructions for $t=2$, (b) providing improved bounds on block length for $t=3$ as well as a general construction for $t=3$ having short block length, (c) providing high-rate constructions for $\left( r=2, \ t \in \{4,5,6,7\} \right) $ and (d) providing short-block-length constructions for general $(r,t)$.  Most of the codes constructed here are binary codes.
\end{abstract}



\section{Introduction}
All codes discussed are linear and over a finite field $\mathbb{F}_q$. Throughout, $[n,k,d_{\min}]$ will denote the block length, dimension and minimum distance of the linear code. Two classes of codes with locality are considered here, namely those that offer parallel and sequential recovery respectively. While much of the discussion holds for general $q$, most of the codes constructed in the paper are binary, corresponding to $q=2$.   The focus in the paper is on the construction of codes having short or minimum block length.   This is motivated in part by practical implementation considerations, and in part, with a view to getting some insight into what is to be gained by increasing the block length. 
 
\subsection{Parallel and Sequential Recovery} \label{sec:two-classes} 

\paragraph{Codes with parallel recovery} Codes in this class can be defined as the nullspace of an $(m \times n)$ parity-check matrix $H$, where each row has weight $(r+1)$ and each column has weight $t$, with $nt=m(r+1)$.  Additionally, if the support sets of the rows in $H$ having a non-zero entry in the $i$th column are given respectively by $S^{(i)}_j, j=1,2,\cdots t$, then we must have that 
\bean
S^{(i)}_j \cap S^{(i)}_l & = &  \{i \} , \forall 1 \leq j \neq l \leq t.
\eean
Thus each code symbol $c_i$ is protected by a collection of $t$ orthogonal parity checks (opc) each of weight $(r+1)$. The parameter $r$ is called the locality parameter and we will formally refer to this class of codes as $(r,t)_{\text{par}}$ codes.  When the parameters $r,t$ are known from the context, we will simply term the code as a code with parallel recovery.  This is because recovery from a set of $t$ erasures can be carried out locally and in parallel. 
\paragraph{Codes with sequential recovery} The requirement of this class of codes is that given any set of $s \leq t$ erased symbols, $\{x_1,...,x_s\}$, there is an arrangement of these $s$ symbols (say) $\{x_{i_1},...,x_{i_s}\}$ such that there are $s$ codewords $\{h_1,...,h_s\}$ in the dual of the code, each of weight $\leq r+1$, with $i_j \in \text{support}(h_j)$ and $\text{support}(h_j) \cap \{i_{j+1},...,i_s\} = \emptyset$, $\forall 1 \leq j \leq s$.  The parameter $r$ is again the locality parameter and we will formally refer to this class of codes as $(r,t)_{\text{seq}}$ codes as it can recover from any $s\leq t$ erasures sequentially and locally, using a set of $s \leq t$ parity checks, each of weight $\leq r+1$, as mentioned in the definition above.  Again when $(r,t)$ are clear from the context, we will refer to a code in this class as a code with sequential recovery.  


\subsection{Background} 

The notion of codes with locality was introduced in \cite{GopHuaSimYek}, see also \cite{PapDim,OggDat}.  The initial focus was on recovery from a single erasure and constructions for codes that accomplish this can be found in \cite{GopHuaSimYek,HuaChenLi,KamPraLalKum,TamBar_Optimal_LRC}.   There are several approaches in the literature to local recovery from multiple erasures.   The approach adopted by the authors of \cite{KamPraLalKum,SonDauYueLi}, is to use a stronger local code with $d_{\min}>2$ to protect the code symbols against multiple erasures.  There is a second class of codes, termed as t-availability codes in which each code symbol is covered by $t$ orthogonal parity checks, but these are only required to have support of size $\leq (r+1)$ as opposed to the strict requirement of $=(r+1)$ discussed here.  Codes with $t$-availability can be found discussed in \cite{ZhaWanGe,HuaYaaUchSie,TamBarFro,KimNamSong,SheFuGua,WanZha_Combinatorial_Repair_locality,TamBar_Optimal_LRC,JuaHolOgg,RawPapDimVis_arxiv,WangZhang_multiple_erasure}.  The sequential approach of recovering from multiple erasures, introduced in \cite{PraLalKum}, can also be found discussed in \cite{SonYue_3_Erasure,RawMazVis,SonYue_Binary_local_repair}. 

\subsection{Our Contributions} 

The focus here is on codes with parallel or sequential recovery, that have short block length.  

\paragraph{Parallel Recovery}   In the case of parallel recovery, we derive a simple lower bound on the block length of this class of codes for given $(r,t)$.  It is then shown that a necessary and sufficient condition for the existence of a minimum-block length ($n_{min}$), parallel-recovery code is the existence of a balanced incomplete block design (BIBD) with parameters governed by $\{n_{min},r,t\}$.  Some examples are noted, including cases in which the codes possess in addition, the highest possible rate for the given $\{n_{min},r,t\}$. 

\paragraph{Sequential Recovery}

Our results here can be broken down according to the value of the parameter $t$:
\begin{itemize}
	\item $t=2$: here we generalize the optimal construction by Prakash et al \cite{PraLalKum} to cover a larger parameter set. Each code constructed here has minimum block length for a given $k,r$ and for the case when $r | 2k$, $k$ being the dimension of the code, a characterization of the class of optimal codes with maximum rate which can sequentially recover from $2$ erasures is provided. 
	\item $t=3$: we derive a lower bound on block length for a given $k,r$ that for $k \leq r^{1.8}$ improves upon an earlier bound by \cite{SonYue_3_Erasure} for binary codes. A general construction of codes is presented with rate $\frac{r}{r+3}$ and short block length ($O(r^{1.5})$) that differs by at most $2$ from the lower bound on block length derived here, over a large parameter range. 
	\item $r=2$ and $t \in \{4,5,6,7\}$: several construction of codes having short block length and high rate are provided. 
	\item General $t$: some general constructions having short block length are presented here.
	\eit 
	Most codes constructed here are binary codes. The results on parallel recovery are presented in Section~\ref{sec:parallel}.  Results on sequential recovery appear in Sections~\ref{sec:2_erasures} to \ref{sec:general_t}.
	
\section{Minimum Block-Length Codes with Parallel Recovery}  \label{sec:parallel} 


Let $H$ be the corresponding $(m \times n)$ parity check matrix of an $(r,t)_{\text{par}}$ parallel-recovery code \codec\ over the field $\mathbb{F}_q$ which includes all the orthogonal parities of all symbols.  Then each column of $H$ has Hamming weight $t$ and each row has weight $(r+1)$. The code ${\cal C}$ may also be viewed as $(d_v=t, d_c=(r+1))$-regular LDPC code.  The corresponding bipartite graph of the code must necessarily have no cycles of length $4$.  

Let $A$ be the $\{0,1\}$ matrix over the reals $\Re$ given by
\bean
a_{ij} & = & \left\{ \begin{array}{rl} 1 & h_{ij} \neq 0, \\   0 & \text{  else  } .  \end{array} \right.
\eean
Given the orthogonal nature of the parity checks and our assumption on the parities involving a code symbol $c_i$, the sum of the inner products between distinct rows of $A$ must satisfy: 
\bean
{m \choose 2} \ \geq \ \sum_{j>i} \left( \sum_{l=1}^na_{i,l}a_{j,l} \right) 
	=  \sum_{l=1}^n \left(\sum_{j>i} a_{il}a_{jl} \right)  =  n {t \choose 2} .
\eean
Using the relation $nt=m(r+1)$, we obtain 
\bea
	m  \geq  (t-1)(r+1) + 1, \ \ \ 
	n \geq  (r+1)^2 - \frac{(r+1)r}{t}.   \label{eq:lb_n} 
\eea
Our interest is in the minimum-block-length case, where \eqref{eq:lb_n} holds with equality and for which a necessary condition is that $t \mid r(r+1)$.   We set $n_{\min}=(r+1)^2 - \frac{(r+1)r}{t}$ and define a code having length $n_{\min}$ to be a minimum-length code. 

When equality holds, it follows that the inner product of every pair of distinct rows of $A$ is exactly equal to $1$.  Let us define the column support sets $B_j \subseteq [m]$ by 
\bean
i \in B_j & \text{ iff } &  a_{i,j}=1 \text{ or equivalently, $h_{ij} \neq 0$}.
\eean
It follows then that the sets $\{ B_j \}_{j=1}^n$ form a $(b,v,\hat{r},\hat{k}, \lambda)$ balanced incomplete block design (BIBD) having parameters 
\bean
b \ = \ n, v \ = \ m, \ \hat{r} = (r+1), \ \hat{k}=t, \ \ \lambda=1 .
\eean
Conversely a BIBD with these parameter values will yield an $(r,t)_{\text{par}}$ block code \codec\ having minimum possible block length.   The rate $R$ of \codec\ clearly satisfies $R \geq 1 -\frac{t}{(r+1)}$.  
%
%
\begin{example}\label{eg:ProjPlane}
	Let $Q=2^s, s \geq 2$ and let $PG(2,Q)$ denote the projective plane over $\mathbb{F}_Q$.  There are $Q^2+Q+1$ points and $Q^2+Q+1$ lines in $PG(2,Q)$.  Each line contains $Q+1$ points and there are $Q+1$ lines through a point.    Set $n=Q^2+Q+1$. Let $H$ be the $(n \times n)$ parity check matrix of a binary code \codec\, i.e., a code over $\mathbb{F}_2$, given by 
	\bean
	h_{ij}= 1 & \Leftrightarrow & \text{ the $i$th point lies on the $j$th line}.
	\eean
	Then it is known that $H$ has rank $3^s+1$ over $\mathbb{F}_2$, and that \codec\ has $d_{\min}=Q+2$, thus \codec\ is a binary $(Q,Q+1)_{\text{par}}$ code having parameters $[Q^2+Q+1,Q^2+Q-3^s,Q+2]$.    A plot of the rate of this code versus the bound by Tamo et al in \cite{TamBarFro} as a function of the parameter $s$ is shown in Fig.~\ref{fig:bibd_TB_bound}.  
\end{example}

While this code is well-known in the literature on LDPC codes, our aim is to draw attention to the fact that this code is a $(Q,Q+1)_{\text{par}}$ code having minimum block length. The parameters of a minimum-block-length code obtained by a similar construction involving lines in the affine plane and $(r,t)=(Q,Q)$ are given by $[n,k,d_{\min}]=[Q^2+Q,Q^2+Q-3^s,\geq Q+1]$, where $Q=2^s, s \geq 2$. 

\begin{conj}[Hamada-Sachar Conjecture (Conjecture 6.9.1 of \cite{AssKey}] 
	Every Projective plane of order $p^s$, p a prime, has p rank at least ${p+1 \choose 2}^s + 1$ with equality if and only if its desarguesian.
\end{conj}
The above conjecture is as yet unproven, but if true, would show that the projective-plane code described in Example \ref{eg:ProjPlane}  would have minimum possible block length and the maximum possible rate over the binary field $\mathbb{F}_2$ among all binary $(Q, Q+1)_{par}$ codes with $n=Q^2+Q+1$ and  $Q=2^s$. 

\begin{example}\label{eg:sts}
	For $t=3$, one can obtain a code by making use of the Steiner Triple System (STS) associated to the point-line incidence matrix of $(s-1)$ dimensional projective space $PG(s-1,2)$ over $\mathbb{F}_2$.   Once again, the rows of $H$ correspond to points in the projective space and the columns to lines.  It follows that $t=3$ and $(r+1)=\frac{2^s-2}{2^2-2} = 2^{s-1}-1$. Let $m=2^s-1$.   It turns out that this yields a binary $(2^{s-1}-2,3)_{\text{par}}$ code having parameters $ [\frac{m(m-1)}{6},\frac{m(m-1)}{6}-m+s,4]$.  A plot comparing the rate of this code and the bound by Tamo et al \cite{TamBarFro} is shown in Fig.~\ref{fig:bibd_TB_bound}.    
	 \end{example} 
\begin{conj}[Hamada's Conjecture (1.91 in \cite{ColDin})]
	The p- rank of any design $D$ with parameters of a geometric design $G$ in PG(n,q) or AG(n,q) $(q=p^m)$ is at least the p-rank of $G$ with equality if and only if $D$ is isomoprhic to $G$.
\end{conj}
	This conjecture has been shown to hold true for the Steiner Triple system appearing in Example \ref{eg:sts}. Thus the code in Example \ref{eg:sts} has, as a binary code, the minimum possible block length and maximum possible rate among all binary $(\frac{m-1}{2}-1,3)_{par}$ codes with $n=\frac{m(m-1)}{6}$ and $m=2^s-1$.

\begin{figure}
	\begin{center}
		\begin{minipage}{3in}
			\begin{center}
				\includegraphics[trim= 1in  1in 1in 1in, width=2in]{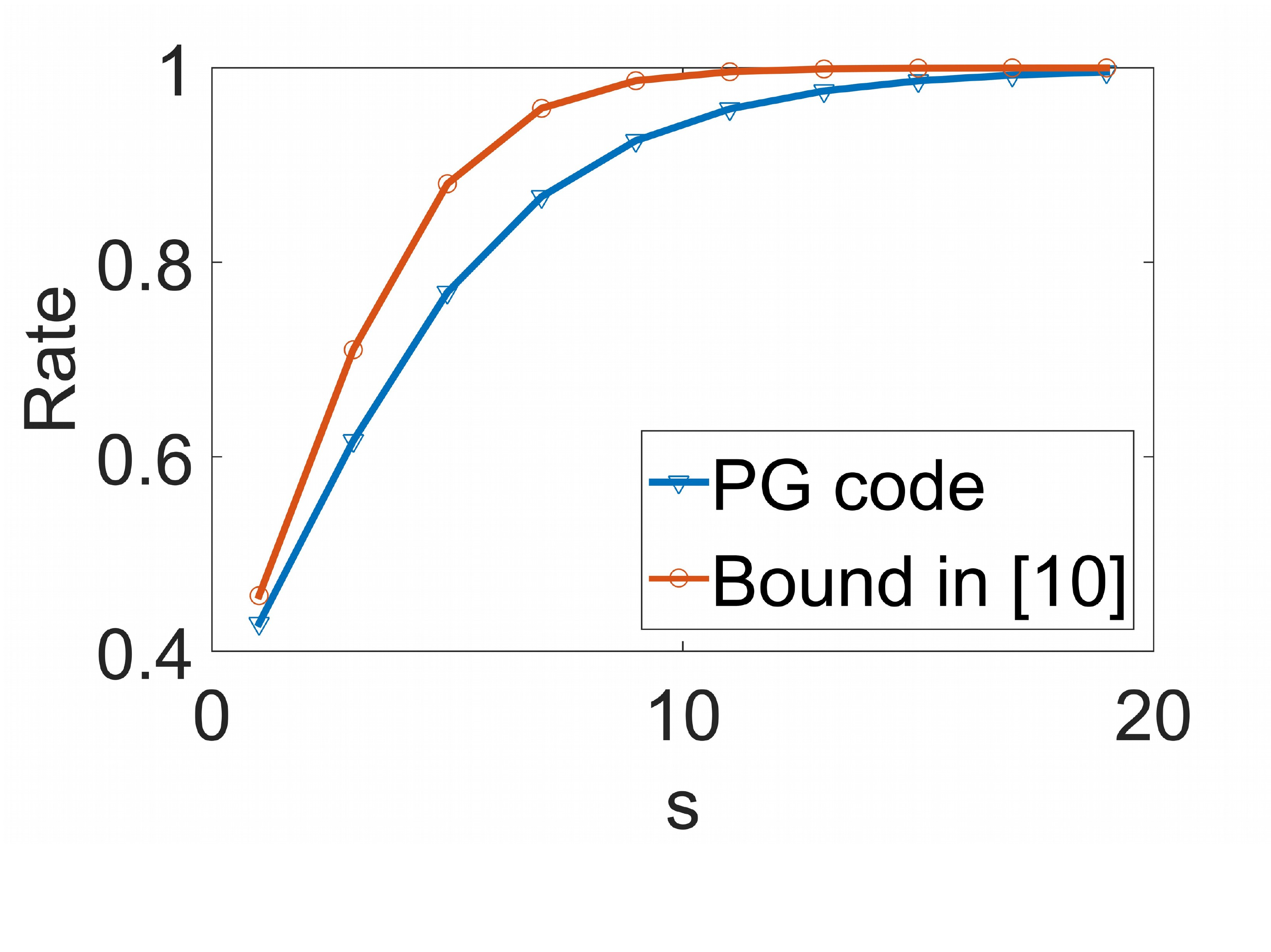}
			\end{center}
		\end{minipage} \hspace{0.5in} 
		\begin{minipage}{3in}
			\begin{center}
				\includegraphics[trim= 1in  1in 1in 1in, width=2in]{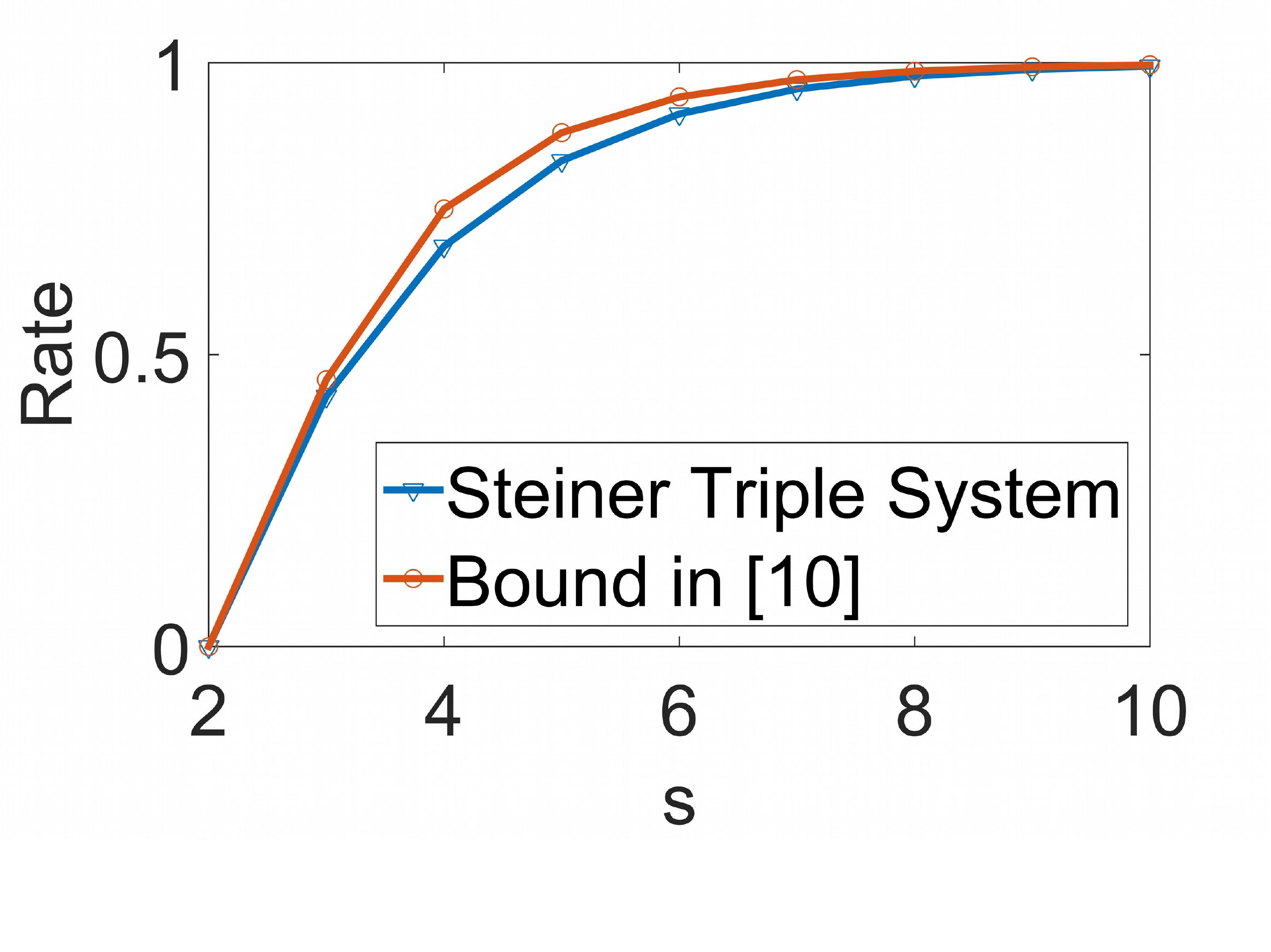}
			\end{center}
		\end{minipage} 
	\end{center}
	\caption{Comparing the rates of the projective plane and Steiner-triple-system-based codes with the bound in \cite{TamBarFro}.}
	\label{fig:bibd_TB_bound}
\end{figure}


\section{Codes with Sequential Recovery from Two Erasures} \label{sec:2_erasures} 

In \cite{PraLalKum}, it is shown that for either sequential or parallel recovery from $t=2$ erasures with locality $r$, the rate of the code is upper bounded by 
\bea
\frac{k}{n} \leq \frac{r}{r+2}  \label{eq:2erasure_rate}
\eea
For given dimension $k$ and locality parameter $r$, this leads to the lower bound $n  \geq  k + \lceil \frac{2k}{r}\rceil$ . The authors in \cite{PraLalKum} provide a construction of optimal codes where equality holds in \eqref{eq:2erasure_rate} .  The construction is based on Turan graphs and holds whenever  $r | 2k,  \text{ and in addition, }  \frac{2k}{r} \ = \ (r+\beta), \text{ for some $\beta | r$  }.$ 
The construction given in \cite{SonYue_3_Erasure} for $t=2$ sequential erasure correction requires $\lfloor\frac{k}{r}\rfloor \geq r$.

In the present paper, we present a simple construction that has minimum possible block length for given $\{r,k\}$ and that holds for a larger set of parameters, as it only requires that 
$\left\lceil \frac{2k}{r}\right\rceil  \geq   r+1 \text{   for } r \mid 2k. \text{ and } \left\lceil \frac{2k}{r}\right\rceil  \geq   r+2 \text{   for } r \nmid 2k$.

\begin{const}   [Sequential recovery from $t=2$ erasures]  \label{const:reg_graph}
	Let $2k=ar+b , 0\leq b \leq r-1$. Let $\mathcal{G}$ be a graph on a set of $m= \lceil \frac{2k}{r} \rceil$ nodes with `$a$' nodes having degree $r$ and an additional node having degree $b$, for the case when $b>0$.  Let each edge in the graph represent an information symbol in the code and each node represent a parity check symbol which corresponds to the sum of the information symbols corresponding to the edges connected to that node. The code is systematic and is defined by the information symbols corresponding to edges and the parity symbols corresponding to the nodes of $\mathcal{G}$. The dimension of this code, i.e., the number of information symbols, is clearly equal to $k$ and block length $n= k + \lceil \frac{2k}{r}\rceil$.   Since each parity check symbol represented by a node is the sum of at most $r$ information symbols, the corresponding parity check involves at most $r+1$ code symbols. Thus the code has locality $r$. It is straightforward to see that the code can sequentially recover from $2$ erasures.  
\end{const}
Noting that the graph $\mathcal{G}$ is regular in the case, $b=0$, we will refer to the code described in Construction \ref{const:reg_graph} with $b=0$ as the \emph{Regular Graph Code}.   The parameter sets $(k,r)$ for which the graph  $\mathcal{G}$ of the form described in Construction~\ref{const:reg_graph} exists can be determined from the Erd\"os-Gallai theorem \cite{ErdosGallai} and the parameters sets turns out to be $\{(k,r): \lceil \frac{2k}{r} \rceil = m\geq r+1 \}$ and $\{(k,r): \lceil \frac{2k}{r} \rceil = m\geq r+2 \}$ , when $b=0$ and $b>0$ respectively. 
	
\subsection{Uniqueness of Rate Optimal Codes for $2$ Erasures}
	In \cite{PraKamLalKum}, Prakash et. al. introduced the class of $(r,\delta)$ codes which have MDS codes having $d_{\min}>2$ as the local codes. We will refer to these codes for better clarity as $(r,\delta)_{MDS}$ codes. In this section, we prove that a rate-optimal sequential code- with-locality for $2$ erasures must have a specific form.
 
 Prakash et. al. \cite{PraLalKum} derived bound \eqref{eq:2erasure_rate} on the rate of a $(r,t=2)_{seq}$ code.

 From the derivation of this rate bound  given in \cite{PraLalKum}, it is straightforward to see that an $[n,k,d]$ code with locality $r$ and 2 sequential erasure correction, which achieves the bound \ref{eq:2erasure_rate} will have a parity check matrix (after possible permutation of code symbols) of the form $[I|H']$, where $I$ denote the identity matrix of order $n-k$ and $H'$ is a $(n-k)\times k$ matrix with all columns having a hamming weight of $2$ and all rows having a hamming weight of $r$.
\begin{thm}
	A code with locality $r$ capable of sequential recovery from $2$ erasures and achieving the rate upper bound in \eqref{eq:2erasure_rate} must fall (upto coordinate permutation) into one of the following classes:
	\ben
	\item A regular graph code (possibly defined over a larger field with coefficients from the larger field in place of $1$'s in the parity check matrix),
	\item A $(r,3)_{MDS}$ code in Prakash et. al. \cite{PraKamLalKum}.
	\item A code that is the direct product ${\cal C} = {\cal C}_1 \times {\cal C}_2$ where ${\cal C}_1$ is a regular graph code (possibly defined over a larger field with coefficients from the larger field in place of $1$'s in the parity check matrix) and ${\cal C}_2$ is a $(r,3)_{MDS}$ code in \cite{PraKamLalKum}.
	\een	
\end{thm}
\begin{proof}
	Assume that $H$ is the parity check matrix of an $[n,k,d]$ code with locality $r$ and sequential repair capability of $t=2$, with rate $\frac{r}{r+2}$. As mentioned before, we can write (after possible permutation of code symbols)
	\beq
	H=[I|H']
	\eeq
	where $I$ denote the identity matrix of order $n-k$ and $H'$ is a $(n-k)\times k$ matrix with all columns having a hamming weight of $2$ and all rows having a hamming weight of $r$.
	Consider $2$ rows of $H$, $R_1$ and $R_2$. Assume that $support(R_1)$ and $support(R_2)$ intersect at columns $C_1,C_2...C_s$. Since all columns in $H'$ has weight exactly $2$, the columns $C_1,C_2...C_s$ will have non zero entries in $R_1$ and $R_2$ only. Let $A$ denote the $2\times s$ sub matrix obtained by considering the rows $R_1$ and $R_2$ and the columns $C_1,C_2...C_s$ only. In order to recover from any instance of $2$ erasures in the symbols corresponding to the columns of $A$, any two columns of $A$ must be linearly independent. Thus the $2\times s$ sub matrix $A$ forms the generator matrix of an $MDS$ code of block length $s$ and dimension 2. This also says that any vector, obtained by a linear combination of the two rows of $A$ will have a hamming weight at least $s-1$.\\
	Let us consider two extreme cases:\\
	\emph{Case 1: $s=r$:}
	This case corresponds to having $(r,3)_{MDS}$ locality introduced by Prakash et .al.\cite{PraKamLalKum} for the set of symbols in $support(R_1) \cup support(R_2)$ \\
	\emph{Case 2: $s\leq1$:} \\
	In this case two parity checks, represented by $R_1$ and $R_2$ have at most one symbol in common. \\
	If these are the only two cases that can occur for any pair of rows $R_i$ and $R_j$ i.e.,$|support(R_i) \cap support(R_j)| \in \{0,1,r\}$ then the set of code symbols can be partitioned into two sets, one set of symbols forming a regular graph code (with possibly higher field coefficients in place of 1's in the parity check matrix) and the other set of symbols forming $(r,3)_{MDS}$ code, with no parities across these 2 sets i.e., the code will be a direct product of regular graph code (with possibly higher field coefficients in place of 1's in the parity check matrix) and $(r,3)_{MDS}$ code (after possible permutation of code symbols). 
	
	Now, we will prove that $1<s<r$ is not possible for any pair of rows, where $s$ denote the size of the intersection of support, of the pair of rows.
	
	Wlog assume that $1<s<r$ for the pair of rows $R_1$ and $R_2$ i.e., $|support(R_1) \cap support(R_2)|=s$. Let $C_i,C_j$ be two columns belonging to the set of $s$ columns where $support(R_1)$ and $support(R_2)$ intersect. Assume that the symbols corresponding to $C_i$ and $C_j$ are erased. In order to sequentially repair these symbols locally, $R_1$ and $R_2$ must linearly combine with some of the remaining rows of $H$ to get a vector $\mathbf{v}$ with the following properties.
	\ben
	\item Hamming weight of $\mathbf{v}$ is less than or equal to $(r+1)$.
	\item $\mathbf{v}$ has a zero in the coordinate corresponding to $C_i$ and a non zero value in the coordinate corresponding to $C_j$, or vice versa.
	\een
	Assume that a linear combination of $l$ rows $\{R_1,R_2,R_3...R_l\}$ results in $\mathbf{v}$. Let $s_{ij}$ denote $|support(R_i)\cap support(R_j)|$. Clearly, $s_{12}=s$. If $s_{ij}>0$, we have shown that the $2\times s_{ij}$ sub matrix formed by the rows $R_i$ and $R_j$ and the columns in $support(R_i)\cap support(R_j)$ form a generator matrix of an MDS code of block length $s_{ij}$ and dimension 2 and they linearly combine to form a vector of hamming weight at least $s_{ij}-1$. Thus the hamming weight of $\mathbf{v}$ is at least $l + \sum_{1 \leq i<j\leq l,s_{ij}>0}^{}(s_{ij}-1) + f$, where the factor of $l$ comes from the identity part of $H$ (i.e., columns $1$ to $n-k$ of $H$) and $f$ comes from the single weight columns in the sub matrix $L$ formed by the rows $\{R_1,R_2,R_3...R_l\}$ and columns $n-k+1$ to $n$ of $H$.
	\begin{align}
		l + \sum_{1 \leq i<j\leq l,s_{ij}>0}^{}(s_{ij}-1) + f &\leq r+1 \notag\\ 
		l + \sum_{1 \leq i<j\leq l,s_{ij}>0}^{}s_{ij} - {l \choose 2} + f &\leq r+1 \notag\\ 
		l + 2 \sum_{1 \leq i<j\leq l,s_{ij}>0}^{}(s_{ij}) - {l \choose 2} + f &\leq r+1 + \sum_{1 \leq i<j\leq l,s_{ij}>0}^{}(s_{ij}) \label{eq:wt1}\\ 
		\text{Also, by counting the non zero entries in $L$ row wise and column wise} \notag \\
		 f+ 2\sum_{1 \leq i<j\leq l,s_{ij}>0}^{}s_{ij} &= lr \label{eq:wt2}\\
		\sum_{1 \leq i<j\leq l,s_{ij}>0}^{}s_{ij} &\leq \frac{lr}{2} \label{eq:wt3}
	\end{align}
	
	substituting \eqref{eq:wt2} and \eqref{eq:wt3} in \eqref{eq:wt1} gives: 
	\bean
	l + lr-{l \choose 2} \leq r+1 + \frac{lr}{2}
	\eean

	simplifying and assuming $l>2$ we get
	\bean
	r \leq l-1
	\eean
	Hence we get $l\geq r+1$, when $l>2$. But when $l\geq r+1$ the coordinates in the identity part (columns $1$ to $n-k$) will add a Hamming weight of $r+1$ to $\mathbf{v}$, making the Hamming weight of $\mathbf{v}$ greater than $r+1$ as it must also have a non zero $C_j^{th}$ or $C_i^{th}$ coordinate. Hence, if $1<s<r$, $l>2$ is not possible. 
	Now, assume $l=2$ i.e., a linear combination of $R_1$ and $R_2$ should give $\mathbf{v}$. The Hamming weight of a linear combination of $R_1$ and $R_2$ is at least $(s-1)+2(r+1-s)$ (weight $s-1$ comes from the  coordinates in $support(R_1)\cap support(R_2)$, and weight $2(r+1-s)$ comes from the remaining coordinates in $support(R_1)\cup support(R_2)$). We need,
	\begin{align*}
		(s-1)+2(r+1-s) &\leq r+1\\
		s &\geq r
	\end{align*}
	which is not possible as  $1<s<r$. Hence $l=2$ is also not possible. Hence putting together we have $l \geq 2$ not possible but for $1<s<r$, but we need to linearly combine $l \geq 2$ rows to get $\mathbf{v}$. Hence $1<s<r$ is not possible.
\end{proof}

\section{Codes with Sequential Recovery from Three Erasures} \label{sec:3_erasures} 
 
We first present an improved lower bound on the block length of binary codes that can sequentially recover from $3$ erasures for $k \leq r^{1.8}-1$.  This is followed by the construction of a short block length code that generalizes an instance of the Turan graph construction in \cite{PraLalKum}.

\subsection{Bound on Minimum Block Length} \label{sec:t_eq_3_bound} 
 
	In \cite{SonYue_3_Erasure}, W. Song et. al. derived the following lower bound on the block length of codes that can sequentially recover from three erasures.
	\bea
	n\geq k+\left\lceil\frac{2k+\lceil\frac{k}{r}\rceil}{r}\right\rceil  \label{eq:song_3}
	\eea
Constructions were also provided of codes meeting the above bound for $\lceil \frac{k}{r} \rceil \geq r$.   Here, we present a new bound on block length. Simulation shows that the new bound is tighter than \eqref{eq:song_3} for $r \leq k \leq r^{\tiny{1.8}}-1$ for $1 \leq r \leq 200$. We also provide a few sporadic examples where this bound turns out to be tight.   For binary codes, our bound takes on the form:
\bea
n \geq k+\min_{s_1} \max \{f_1(s_1), f_2(s_1), s_1 \} ,\label{eq:min_len}
\eea
\bean 
\text{where}\hspace{30pt} f_1(s_1) =   \left\lceil \frac{-(2r-5)+\sqrt{(2r-5)^2 +4(6k+s_1^2-5s_1)}}{2} \right\rceil \\
f_2(s_1) = \left\lceil \frac{-(4r-4+2s_1)+\sqrt{(4r-4+2s_1)^2 +4(12k+3s_1^2-4s_1-7)}}{2} \right\rceil .
\eean
\begin{proof}
	Let $H'$ be the parity check matrix of an $(r,3)_{seq}$ code $\cal{C}$ (possibly having global parities), with $m'$ linearly independent rows, block length $n$ and dimension $k$. 
	Let 
	\bean
	   B_0 = span(h : weight(h) \leq r+1, h \in \cal{C^{\perp}})
	\eean
	Let $\{c_1,...,c_m\}$ be a basis of $B_0$ with $weight(c_i) \leq r+1$. Let
	\bean
	H & = & 
	\left[ \begin{array}{c} 
		c_1  \\
		\vdots \\
		c_m   \end{array} \right].
	\eean
	Let's extend the basis $\{c_1,...,c_m\}$ of $B_0$ to a basis of $\cal{C^{\perp}}$ and form the parity check matrix $H'$ of $\cal{C}$ with this basis of $C^{\perp}$. Hence 
	\bean
	H' & = & 
	\left[ \begin{array}{c} 
		H  \\
		H_1   \end{array} \right].
	\eean
	where $H_1$ contains the extra vectors coming from extending the basis $\{c_1,...,c_m\}$ of $B_0$ to a basis of $\cal{C^{\perp}}$. The number of row vectors in the above matrix must be $m'$.
	Since $k$ is the dimension of the code $\cal{C}$ with parity check matrix $H'$ then $n=k+m' \geq k+m$. Now $H$ is a parity check matrix of a $(r,3)_{seq}$ code with same block length $n$, as its row span has all the parities of weight $\leq r+1$ of $\cal{C^{\perp}}$ defined by $H'$.
	Now we consider the code defined by the parity check matrix $H$ with $m$ linearly independent rows which is a $(r,3)_{seq}$ code and derive a lower bound on $m$ as a function of $k$ and $r$. Using $n=k+m' \geq k+m$ and the derived lower bound on $m$, we get a lower bound on $n$.
	
	Let $s_1,s_2$ be the number of columns of $H$ with weights $1$ and $2$ respectively. Then by simple counting of non zero entries of $H$ row wise and column wise:
	\bea
	s_1+2 s_2 +3(n-s_1-s_2) \leq m(r+1) \notag \\
	3n-m(r+1)-2s_1 \leq s_2 \label{eq:bound_s2}
	\eea
	Permute the columns and rows of $H$ matrix such that :\\
	\bean
	H & = & 
	\left[ \begin{array}{cc} 
		D_{s_1} & A  \\
		0 & B   \end{array} \right].
	\eean
	
	where $D_{s1}$ is a diagonal matrix of order $s_1$ with non zero diagonal entries. Now the $s_2$, two-weight columns are to the right of $D_{s_1}$. In these $s_2$ columns, we cannot have a column with 2 non zero entries in the first $s_1$ rows, as this would imply $d_{min}<=3$ (where $d_{min}$ is the minimum distance of the code defined by the parity check matrix $H$) as the code defined by parity check matrix $H$ is also a $(r,3)_{seq}$ code and hence $d_{min} \geq 4$. Hence :\\
	Let $f_1$ = number of columns of weight $2$ with exactly one non zero entry in the first $s_1$ rows.\\
	$f_2$ = number of columns of weight $2$ with both non zero entries in the last $m-s_1$ rows.\\
	\begin{align*}
	s_2 &=f_1+f_2 \\
	f_1 &\leq s_1(m-s_1) \\
	f_2 &\leq N(m-s_1,2,4)
	\end{align*}
	where $N(m-s_1,2,4)$ is the maximum number of columns in a parity check matrix with $m-s_1$ rows and column weight 2, of a code with $d_{min}>=4$.
	Restricting to binary codes, it is straightforward to see that,
	$N(m-s_1,2,4) \leq {m-s_1 \choose 2}$.
	With a little bit of thought this can be tightened to:\\
	\begin{align}
	N(m-s_1,2,4) &\leq \frac{(m-s_1+3)(m-s_1+1)}{4}+1\notag\\
	\text{Hence,  } \
	s_2  &\leq s_1(m-s_1) + {m-s_1 \choose 2} \label{eq:bin_1}\\
	s_2  &\leq s_1(m-s_1) + \frac{(m-s_1+3)(m-s_1+1)}{4}+1 \label{eq:bin_2}
	\end{align}
	
	Hence substituting both the above bounds \eqref{eq:bin_1},\eqref{eq:bin_2} on $s_2$ in \eqref{eq:bound_s2}:\\
	\begin{align}
	3n-m(r+1)-2s_1  &\leq s_1(m-s_1) + {m-s_1 \choose 2} \label{eq:binary_1}\\
	3n-m(r+1)-2s_1  &\leq s_1(m-s_1) + \frac{(m-s_1+3)(m-s_1+1)}{4}+1 \label{eq:binary_2}
	\end{align}

	\eqref{eq:binary_1} (On using $n \geq k+m$) leads to:
	\bean
	m^2 + m(2r-5)-(6k+s_1^2-5s_1) \geq 0
	\eean
	which gives:\\
	\bean
	m \geq \lceil \frac{-(2r-5)+\sqrt{(2r-5)^2 +4(6k+s_1^2-5s_1)}}{2} \rceil=f_1(s_1) \\
	\geq \lceil \frac{-(2r-5)+\sqrt{(2r-5)^2 +4(6k-6)}}{2} \rceil
	\eean
	which when added with $k$ gives a better lower bound on $n$ than\eqref{eq:song_3}, over some parameter range in $k \leq r^2$. 
	
	\eqref{eq:binary_2} (On using $n \geq k+m$) leads to:
	\bean
	m^2 + m(4r-4+2s_1)-(12k+3s_1^2-4s_1-7) \geq 0
	\eean
	which gives:\\
	\begin{align*}
	m &\geq \left\lceil \frac{-(4r-4+2s_1)+\sqrt{(4r-4+2s_1)^2 +4(12k+3s_1^2-4s_1-7)}}{2} \right\rceil \\
	&=f_2(s_1) 
	\end{align*}
	Taking $m \geq min_{s_1} max(f_1(s_1),f_2(s_1),s_1)$ and using $n \geq k+m$, we get 
	\begin{align*}
	n \geq k+min_{s_1} max(f_1(s_1),f_2(s_1),s_1)
	\end{align*}
	Simulation shows that this bound is tighter than Song et. al. bound \eqref{eq:song_3} for $k\leq r^{1.8}-1$ for $1 \leq r \leq 200$.
\end{proof}
We provide some examples which achieve the bound \eqref{eq:min_len}. 
\ben
\item When $r=4,k=8,t=3$, our bound \eqref{eq:min_len} gives $n \geq 14$ whereas the bound in  \eqref{eq:song_3} gives $n \geq 13$.  The binary code associated to the parity-check matrix shown below achieves our tighter bound and hence represents a code of minimum possible block length for $k=8,t=3$. 
\bea
H  =  \left[  \begin{array}{cccccc|cccccccc}
		1 & 0 & 0 & 0 & 0 & 0 & 1 & 1 & 1 & 1 & 0 & 0 & 0 & 0 \\ 
		0 & 1 & 0 & 0 & 0 & 0 & 0 & 0 & 0 & 0 & 1 & 1 & 1 & 1 \\ 
		0 & 0 & 1 & 0 & 0 & 0 & 1 & 1 & 0 & 0 & 1 & 1 & 0 & 0 \\ 
		0 & 0 & 0 & 1 & 0 & 0 & 0 & 0 & 1 & 1 & 0 & 0 & 1 & 1 \\ 
		0 & 0 & 0 & 0 & 1 & 0 & 1 & 0 & 1 & 0 & 1 & 0 & 1 & 0 \\ 
		0 & 0 & 0 & 0 & 0 & 1 & 0 & 1 & 0 & 1 & 0 & 1 & 0 & 1
		\end{array} \right] 
		\label{eq:n_eq_14_code} 
\eea
This code is an example of a more general construction presented below in the next subsection (see Fig.~\ref{fig:turan3}). 
\item $n=28,r=7,k=20,t=3$ : (\eqref{eq:song_3} gives $n \geq 27$ for $k=20,r=7$. Our bound \eqref{eq:min_len} gives $n \geq 28$ for $k=20,r=7$. Hence the binary code associated with the parity check matrix given below has the least block length for a binary code for $k=20,r=7$.)
\bean
H = 
\left[ \begin{array}{cccccccccccccccccccccccccccc} 
		1 & 0 & 0 & 0 & 0 & 0 & 0 & 1 & 0 & 0 & 0 & 0 & 0 & 0 & 0 & 1 & 0 & 1 & 0 & 0 & 1 & 1 & 0 & 0 & 0 & 0 & 1 & 1 \\
		0 & 1 & 0 & 0 & 0 & 0 & 0 & 0 & 1 & 0 & 0 & 0 & 0 & 0 & 1 & 0 & 1 & 0 & 0 & 1 & 0 & 1 & 1 & 0 & 0 & 0 & 0 & 1 \\
		0 & 0 & 1 & 0 & 0 & 0 & 0 & 0 & 0 & 1 & 0 & 0 & 0 & 0 & 0 & 1 & 0 & 0 & 1 & 0 & 1 & 1 & 1 & 1 & 0 & 0 & 0 & 0 \\
		0 & 0 & 0 & 1 & 0 & 0 & 0 & 0 & 0 & 0 & 1 & 0 & 0 & 0 & 1 & 0 & 0 & 1 & 0 & 1 & 0 & 0 & 1 & 1 & 1 & 0 & 0 & 0 \\
		0 & 0 & 0 & 0 & 1 & 0 & 0 & 0 & 0 & 0 & 0 & 1 & 0 & 0 & 0 & 0 & 1 & 0 & 1 & 0 & 1 & 0 & 0 & 1 & 1 & 1 & 0 & 0 \\
		0 & 0 & 0 & 0 & 0 & 1 & 0 & 0 & 0 & 0 & 0 & 0 & 1 & 0 & 0 & 1 & 0 & 1 & 0 & 1 & 0 & 0 & 0 & 0 & 1 & 1 & 1 & 0 \\
		0 & 0 & 0 & 0 & 0 & 0 & 1 & 0 & 0 & 0 & 0 & 0 & 0 & 1 & 1 & 0 & 1 & 0 & 1 & 0 & 0 & 0 & 0 & 0 & 0 & 1 & 1 & 1 \\
		0 & 0 & 0 & 0 & 0 & 0 & 0 & 1 & 1 & 1 & 1 & 1 & 1 & 1 & 0 & 0 & 0 & 0 & 0 & 0 & 0 & 0 & 0 & 0 & 0 & 0 & 0 & 0 \\
	\end{array}  \right].
	\eean
\item	$n=10,r=3,k=5,t=3$ : (\eqref{eq:song_3} gives $n \geq 9$ for $k=5,r=3$. Our bound \eqref{eq:min_len} gives $n \geq 10$ for $k=5,r=3$. Hence the binary code associated with the parity check matrix given below has the least block length for a binary code for $r=3,k=5$.)
	\bean
	H & = & 
	\left[ \begin{array}{cccccccccc} 
		1 & 0 & 0 & 0 & 1 & 0 & 0 & 0 & 1 & 1 \\
		0 & 1 & 0 & 0 & 0 & 1 & 0 & 0 & 1 & 0 \\
		0 & 0 & 1 & 0 & 0 & 0 & 1 & 0 & 1 & 1 \\
		0 & 0 & 0 & 1 & 0 & 0 & 0 & 1 & 0 & 1\\
		0 & 0 & 0 & 0 & 1 & 1 & 1 & 1 & 0 & 0\\
	\end{array} \right].
	\eean
	\een
\subsection{A Hypergraph-Based Construction for $t=3$} \label{sec:turan}
	
The construction of a binary code presented below may be viewed as a generalization of an instance of the Turan-graph-based construction presented in \cite{PraLalKum} for the sequential recovery of $2$ failed nodes.
	
\begin{const} \label{const:turan} 
Set $b=3\beta$ for some parameter $\beta \geq 1$.   Let ${\cal G}$ be a hypergraph on $b$ nodes constructed by first partitioning the $b$ nodes into three subsets of nodes, labelled as $A_1,A_2,A_3$, each having $\beta$ nodes.  Next, for every triple of nodes $n_1,n_2,n_3$, $n_i \in A_i$, a hyperedge is placed that connects these three nodes.  Thus there are a total of $\beta^3$ hyperedges.  

Each hyperedge is then associated to a unique message symbol and each of the nodes in the node subsets $A_i$ is associated with a parity-check symbol. The parity-check symbol associated to a node $n_i$ is the sum of all the message symbols associated with the hyperdges connected to $n_i$. The code is defined by the set of message symbols assosiated with hyperedges and the parity check symbols assosiated with the nodes of ${\cal G}$. Thus the code has dimension $k=\beta^3$ and block length $n=3\beta+\beta^3$.  It can be shown that this code has minimum distance $4$, hence the code \codec\ has parameters $[n=\beta^3+3\beta, k=\beta^3, d=4]$.  It follows that the code has rate $\frac{\beta^3}{\beta^3+3\beta} = \frac{r}{r+3}$.

\begin{claim}
	Construction \ref{const:turan} gives a $(r=\beta^2,3)_{seq}$ code with $[n=\beta^3+3\beta, k=\beta^3, d=4]$.
\end{claim}
\begin{proof}
	The proof proceeds by showing that for every instance of $3$ code symbol erasures, there is at least one parity check of weight $r+1$ whose support contains exactly one of the erased symbols, and hence the corresponding symbol can be recovered and subsequently the remaining two symbols can also be recovered using local parities.
	Assume that $x$ number of information symbols are erased and $3-x$ parity symbols are erased. We consider the cases $x=0,x=1,x=2,x=3$ and prove each case.\\
	\emph{case 1 : $x=0$}\\
	Each parity symbol is calculated from $r=\beta^2$ information symbols associated with it. Hence any number of erasures among the parity symbols can be repaired locally if none of the information symbols are erased.\\	
	\emph{Case 2 : x = 1}\\
	Each hyper edge is connected to 3 parity nodes. At most two parity nodes are
	allowed to fail in this case. So the erased information symbol can be recovered using the third parity symbol which is not erased.\\	
	\emph{Case 3 : x = 2}\\
	Let $S_i$ denote the set of nodes $\{n^{(i)}_1,n^{(i)}_2,n^{(i)}_3\}$ which are connected by the $i^{th}$ hyperedge. It is easily checked that $\mid S_i \cap S_j \mid \leq 2 \text{ and therefore } \mid S_i \cup S_j \mid \geq 4 \text{ for all } 1\leq i\neq j \leq \beta^3$. Assume that $i$ and $j$ are the hyperedges representing the erased information symbols $I_i$ and $I_j$ respectively. Let the set $T_{ij}$ be defined as:
	\bean
	T_{ij} & = &\{S_i \cup S_j\} \backslash \{S_i \cap S_j\}\\
	\mid T_{ij} \mid & \geq & 2
	\eean	
	The support of each of the parity checks associated with the parity symbols represented by the nodes in $T_{ij}$ will contain either $I_i$ or $I_j$, but not both. Since $\mid T_{ij} \mid \geq 2$, there are atleast two parity checks whose support contains only one of the erased symbols $I_i$ and $I_j$. The third erasure in this case (which is a parity symbol) can affect at most one of these parity checks. Therefore at least one of the information symbols can be recovered, and subsequently the remaining symbols can also be recovered.\\	
	\emph{Case 4 : $x=3$}\\
	Assume that the information symbols $I_i,I_j$ and $I_k$ are erased, and $i,j,k$ are the corresponding hyperedges.\\	
	Let $S_i$ and $T_{ij}$ be as defined above. 
	Consider $I_i$ and $I_j$. If $\mid S_i \cap S_j \mid < 2$ then $|T_{ij}|\geq 4$. i.e. there are at least $4$ parity checks whose support contains only one of the erased symbols $I_i$ and $I_j$, and hence can be used to recover the corresponding symbol. The third erasure in this case can affect at most three of these parity checks, since the hyperedge corresponding to the information symbol $I_k$ is connected to exactly $3$ nodes. Therefore at least one of the information symbol can be recovered, and subsequently the remaining symbols can also be recovered. 
	
	Consider the case when  $\mid S_i \cap S_j \mid = 2$. Let $S_i \cap S_j = \{n_1,n_2\}$. Clearly, $n_1$ and $n_2$ belong to two different partitions, say $A_1$ and $A_2$. Hyperedges $i$ and $j$ will be connected to two distinct nodes in the third partition $A_3$. Exactly one information symbol belonging to the support of the parity checks associated to the parity symbols represented by these two nodes in $A_3$ has been erased. The third erasure in this case can affect at most one parity check among those parity checks associated with the nodes in $A_3$(Since each hyperedge is connected to exactly one node, in one partition). Hence at least one of the information symbols $I_i$ and $I_j$ can be recovered using the parity check associated to the corresponding node in $A_3$. Therefore sequential recovery is possible in this case.
\end{proof}

 Fig.~\ref{fig:turan3} shows an example construction for the case when $\beta=2$ and hence with parameters $[n=14, k=8, d_{\min}=4]$ with locality parameter $r=4$ and permits sequential recovery from $t=3$ erasures. The parity-check matrix of this code was presented earlier in Section~\ref{sec:t_eq_3_bound} as an example of a code that achieves bound \eqref{eq:min_len} for $t=3$ and appears in \eqref{eq:n_eq_14_code}.
\begin{figure}
	\begin{center}
		\begin{minipage}{2in}
			\begin{center}
				\includegraphics[trim= 0.5in  0.5in 1in 1.5in, width=2in]{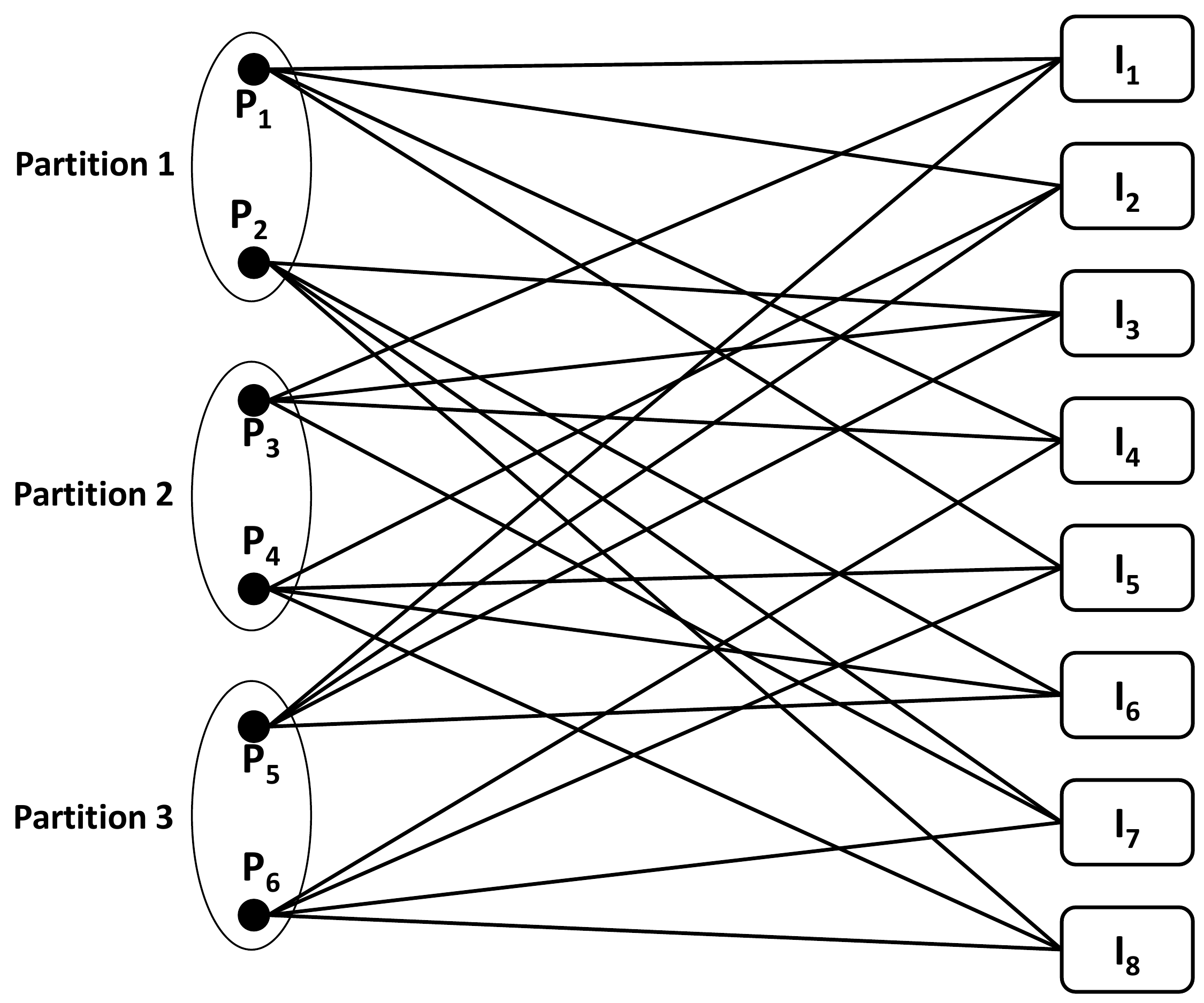}
			\end{center}

		\end{minipage} \hspace{2in} 
		\begin{minipage}{2in}
			\begin{center}
				\includegraphics[trim= 1.4in  1.2in 1in 1.5in, width=2in]{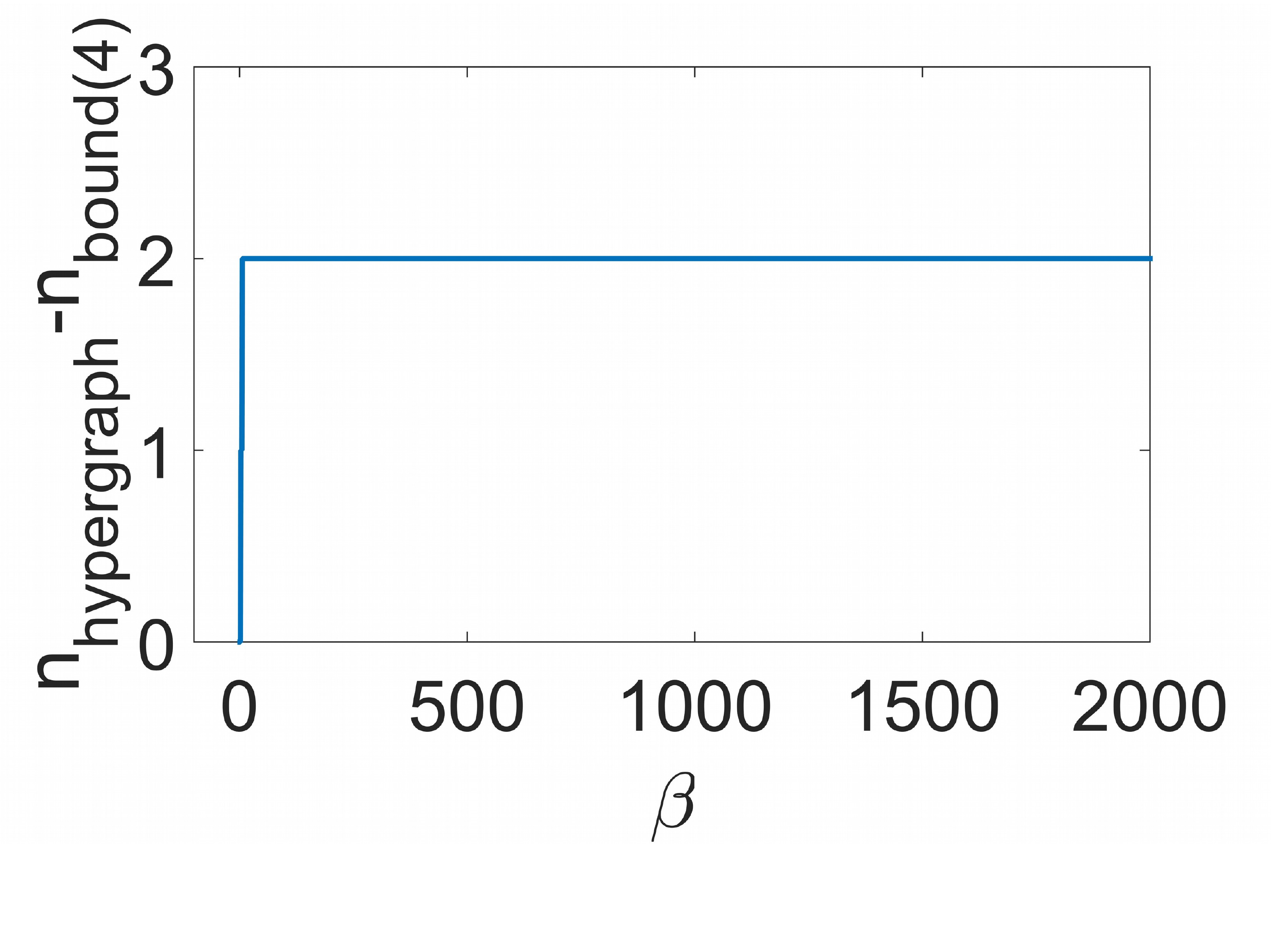}
			\end{center}
		\end{minipage} 
	\end{center}
	\caption{On the left, an example of the hyperedge-based construction given in Construction~\ref{const:turan}.  In the bipartite graph shown here, each node on the right represents a hyperedge and hence, a distinct message symbol.   Each node $P_i$ on the left, represents a parity check symbol.  Thus this code has block length $n=14$ and $k=8$. \\ The plot on the right shows that there is at most a difference of $2$ between the block length of Construction~\ref{const:turan} and the lower bound on block length given by \eqref{eq:min_len} for $1\leq \beta \leq 2000$.}
	\label{fig:turan3}
\end{figure}
\end{const}

\begin{note}
The rate-optimal construction given in \cite{SonYue_3_Erasure} for $3$ erasures requires $\lceil \frac{k}{r} \rceil \geq r$. The hypergraph construction in Construction~\ref{const:turan} described above on the other hand, has a much smaller value of this ratio, namely $\frac{k}{r} = \sqrt{r}$.  Furthermore, the difference between the block length of  Construction~\ref{const:turan} and the bound \eqref{eq:min_len} is $\leq 2$ for $1 \leq \beta \leq 2000$. Thus, Construction~\ref{const:turan} yields high-rate binary codes for parameter sets outside the range of the constructions appearing in \cite{SonYue_3_Erasure}. 
\end{note}

 \section{High Rate Constructions with $r=2$ for Sequential Recovery} \label{sec:r_eq_2}
	
In \cite{TamBarFro}, Tamo et al. derived the upper bound \ref{tamo_rate} 
	\bea
	\frac{k}{n}\leq\frac{1}{\prod_{j=1}^{t}(1+\frac{1}{jr})} \label{tamo_rate}
	\eea
on the rate of codes with availability. However, to the author's knowledge, prior to this paper, there were no general constructions, either sequential or parallel, that achieved this bound for $r=2,t=4,5,6$. 

In this context, the constructions for $(r,t)_{seq}$ codes presented below for $r=2$ and $t=5,6,7$ achieve a larger rate than what appears on the right side of \eqref{tamo_rate}. For $t=4$, we provide a construction having rate close to the bound \eqref{tamo_rate}. The rate of the bipartite graph based constructions by Rawat et.al.(Section VI-A of \cite{RawMazVis}) is $\frac{r-1}{r+1}=0.33$, for $r=2$. Our codes achieve improved rates compared to  $\frac{r-1}{r+1}$ for $t=4,5,6$ and an equal rate for $t=7$. However, for $t=7$, the construction presented here has a smaller block length. The constructions will be presented in graphical form.  In all of the graphs, each node represents a code symbol and a parent node stores the parity of its children. Throughout this section we will use the terms 'nodes' and 'code symbols' synonymously and refer to the code symbols using the same labels as the nodes representing them.

	\subsection{Construction for $t=4,r=2$} \label{sec:t_eq4_r_eq2}
	The construction below yields a systematic $(r=2,t=4)_{seq}$ code with dimension $k$. 
	Let $k=4l,l>1$. Arrange $k$ nodes $I_1...I_k$, representing information symbols, as shown in Fig \ref{fig:four_erasure}. Construct $k$ parity nodes $P_1...P_k$ as shown in Fig \ref{fig:four_erasure} where $P_i$ is the parity of $\{I_i, \ I_{i+1}\}$, for $i=1...k-1$ and $P_k$ is the parity of $\{I_1, I_k\}$. Add a second layer of $k/2$ parities $Q_1...Q_{(k/2)}$.  $Q_i$ is the parity of $\{ P_i, P_{i+\frac{k}{2}}\}, \ i=1\dots \frac{k}{2}$.  This code has rate $\frac{k}{k+k+(k/2)} =0.4 $.
	
	\begin{figure}[ht!]
		\centering
		\includegraphics[height=2in,trim=0 0.2in 0 0.2in]{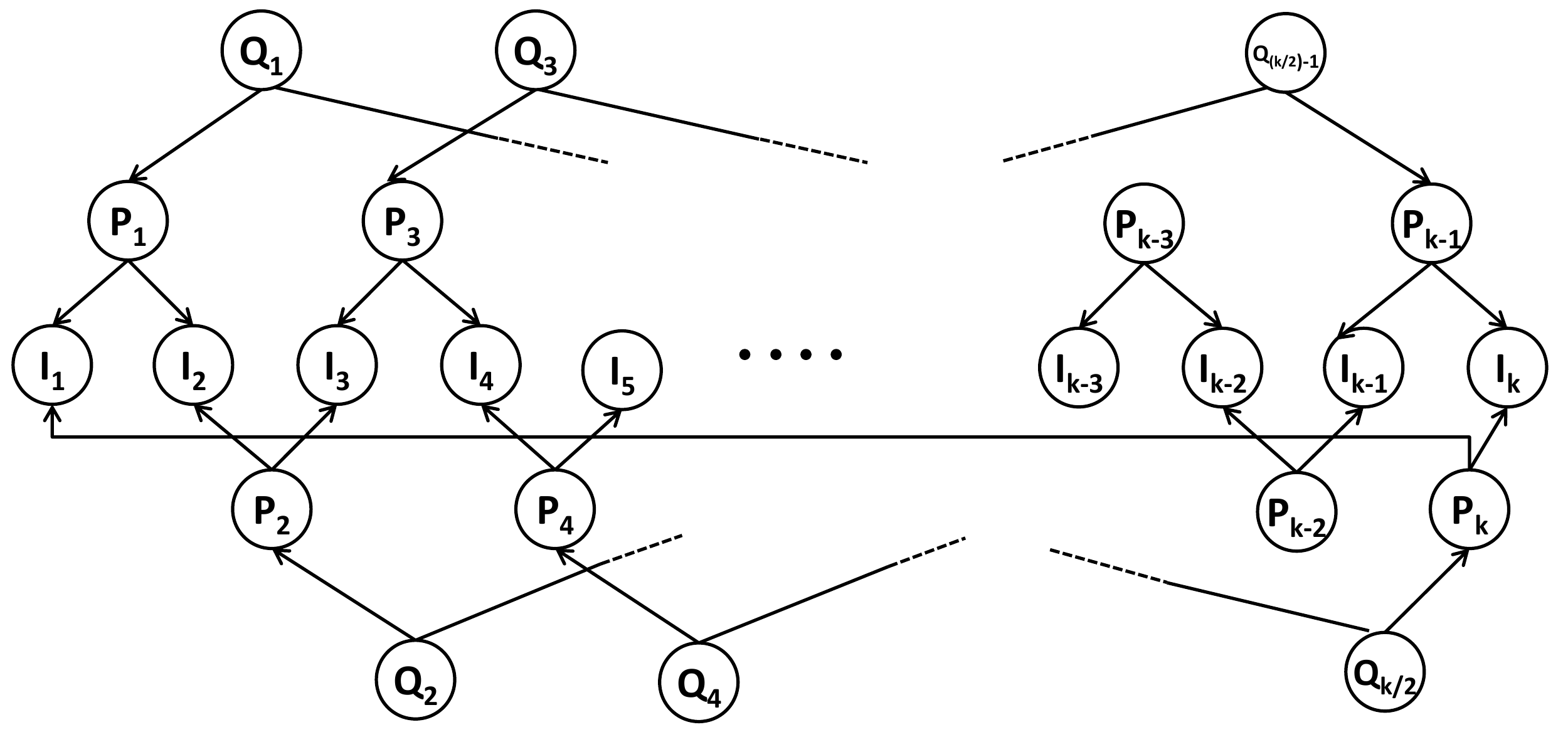}
		\caption{The Four Erasure Correcting Code}
		\label{fig:four_erasure}
	\end{figure}
	\begin{claim}
		The construction defined in section \ref{sec:t_eq4_r_eq2} generates a $(2,4)_{seq}$ code.
	\end{claim}
	\begin{proof}
		Assume that $x$ number of information symbols are erased and $4-x$ parity symbols are erased. We consider the cases $x=0,x=1 \dots x=4$ and prove each case.\\
		Here, we consider only representative worst case scenarios. Remaining cases can be analyzed similarly.\\
		\emph{case 1 : $x=0$}\\
		All parities are derived from information symbols. Hence each erased parity symbol can be recovered locally using existing parity/information symbols.\\
		\emph{Case 2 : $x=1$ }\\
		WLOG assume that node $I_1$ was erased. If either $P_1$ or $P_k$ is not erased, then $I_1$ can be recovered. Therefore assume that both $P_1$ and $P_k$ have been erased. $P_1$ can be recovered using $Q_1$ \& $P_{1+\frac{k}{2}}$. Similarly, $P_k$ can be recovered using $Q_\frac{k}{2}$ and $P_\frac{k}{2}$. But we can erase only one more parity symbol. Hence, either $P_1$ or $P_k$ can be recovered and subsequently all the remaining $3$ nodes can be recovered.\\
		\emph{Case 2 : $x =2$}\\
		WLOG assume that node $I_1$ was erased. It has $2$ repair sets $\{P_1,I_2\}$ and $\{P_k,I_k\}$ which can be used to recover $I_1$.\\
		\emph{Case 2.1 :} Assume that both $P_1$ and $P_k$ are erased. Both of them can be recovered using the second layer of parity symbols ($Q_i$'s). $I_1$ can be recovered since at most  one of the symbols $I_k$ and $I_2$ is allowed to fail.\\
		\emph{Case 2.2 :} Assume that $P_1$ and $Q_1$ are erased. $I_1$ can be recovered using the set $\{P_k,I_k\}$. Hence assume that $I_k$ is erased. But $I_k$ can be recovered using $P_{k-1}$ and $I_{k-1}$. Subsequently all the remaining erased nodes can be recovered.\\
		\emph{Case 2.3 :} Assume that $P_1$ are erased. $I_1$ can be recovered using the set $\{P_k,I_k\}$. Since erasure of $P_k$ is already handled in Case 2.1, assume that $I_k$ is erased. But $I_k$ can be recovered using $P_{k-1}$ and $I_{k-1}$. Hence assume that $P_{k-1}$ is erased. Now $P_1$,$P_{k-1}$ can be recovered from higher layer parities $Q_i's$ since $k>4$. Subsequently all the remaining erased nodes can be recovered.\\
		\emph{Case 3 : $x=3$}\\
		\emph{Case 3.1 : }Assume that $P_1$ was erased.\\
		But $P_1$ can be recovered using $Q_1$ and $P_{1+\frac{k}{2}}$. The three information symbols can be recovered due to similar reasoning as in case 4 below.\\
		\emph{Case 3.2 : }Assume that $Q_1$ was erased.\\
		But $Q_1$ can be recovered using $P_1$ and $P_{1+\frac{k}{2}}$. The three information symbols can be recovered due to similar reasoning as in case 4 below.\\
		\emph{Case 4 : $x=4$ }\\
		Assume that $I_1$ is erased. There are $2$ sets $\{P_1,I_2\}$ and $\{P_k,I_k\}$ which can be used to recover $I_1$. Assume that both nodes $I_2$ and $I_k$ have been erased so that both these sets cannot be used to recover $I_1$. For recovering $I_2$ the set of code symbols $\{P_2,I_3\}$  can be used and for recovering $I_k$ the set $\{P_{k-1},I_{k-1}\}$ can be used. Since we are allowed to erase at most one more information symbol, either $I_2$ or $I_k$ can be recovered. Rest of the 3 information symbols can be recovered due to similar reasoning.
		
	\end{proof}

	The following constructions use a similar approach to generate codes for $t=5,6 \text{ and } 7,r=2$. The proof of sequential recovery for $t=5,6,7,r=2$ is similar to the $t=4,r=2$ case given above.

	\subsection{Construction for $t=5,r=2$}
	A $(2,5)_{seq}$ code can be constructed from the $(2,4)_{seq}$ code constructed above, by adding additional parity symbols $R_1 \dots R_{(k/8)}$. $R_i$ is the parity of $\{ Q_{2i-1}, Q_{2i-1+\frac{k}{4}}\}$. Here we need the additional requirement that $8|k$.   
	
	\begin{figure}[ht!]
		\centering
		\includegraphics[height=2in]{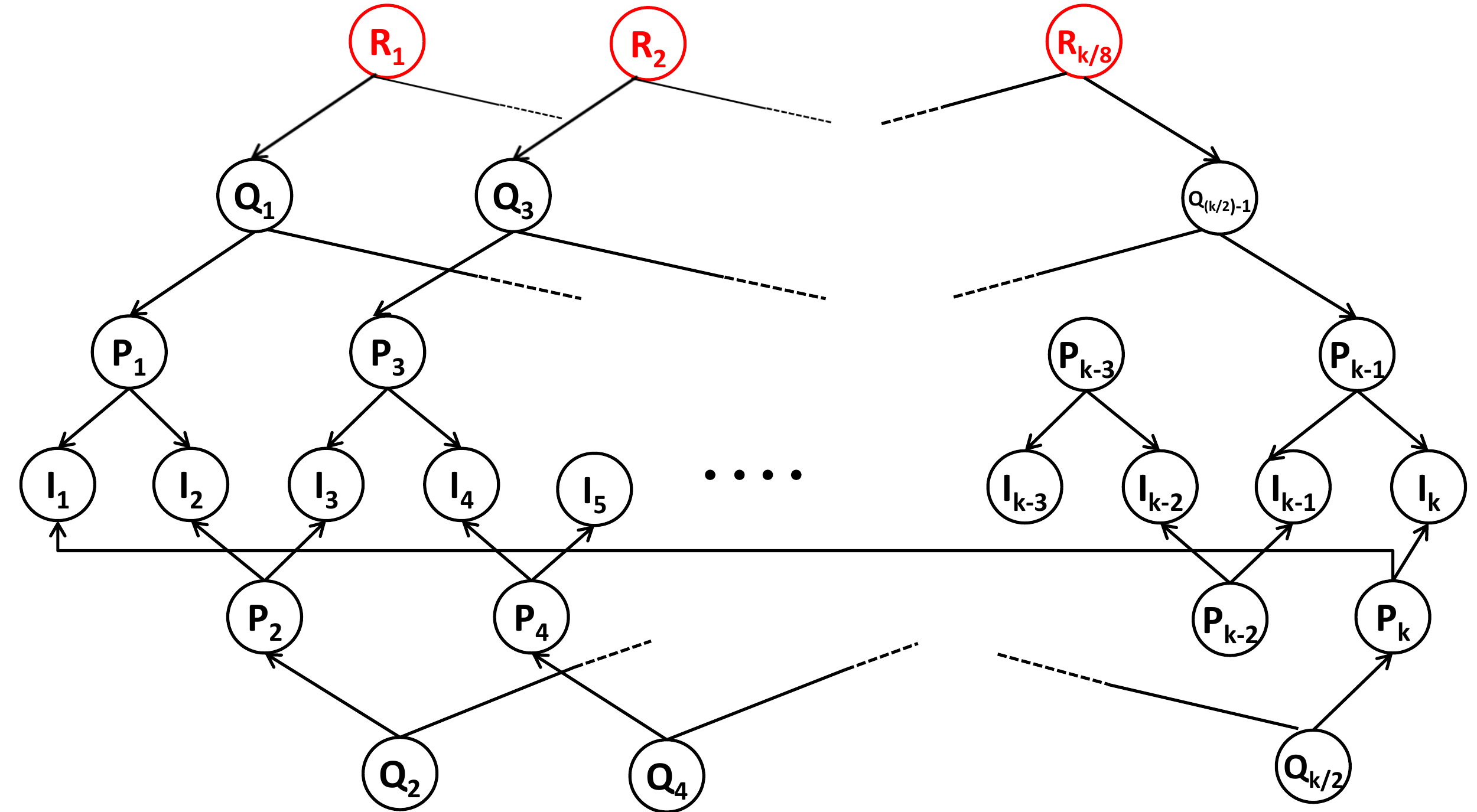}
		\caption{The Five Erasure Code}
		\label{fig:five_erasure}
	\end{figure}
	
	\begin{align*}
		\text{Rate of the code } &= \frac{k}{k+k+(k/2)+(k/8)}\\
		&=0.3810
	\end{align*}
	
	\subsection{Construction for $t=6,r=2$}
	A $(2,6)_{seq}$ code can be constructed from the $(2,5)_{seq}$ code constructed as described above, by adding additional parities $S_1 \dots S_{(k/8)}$ and $T_1...T_{k/8}$. $S_i$ is the parity of $Q_{2i}$ and $Q_{2i+(k/4)}$.  $T_i$ is the parity of $P_{(4i-2)}$ and $P_{4i}$. Here we need the additional requirement of $8|k$ and $k\geq16$.
	\begin{figure}[ht!]
		\centering
		\includegraphics[height=3.5in]{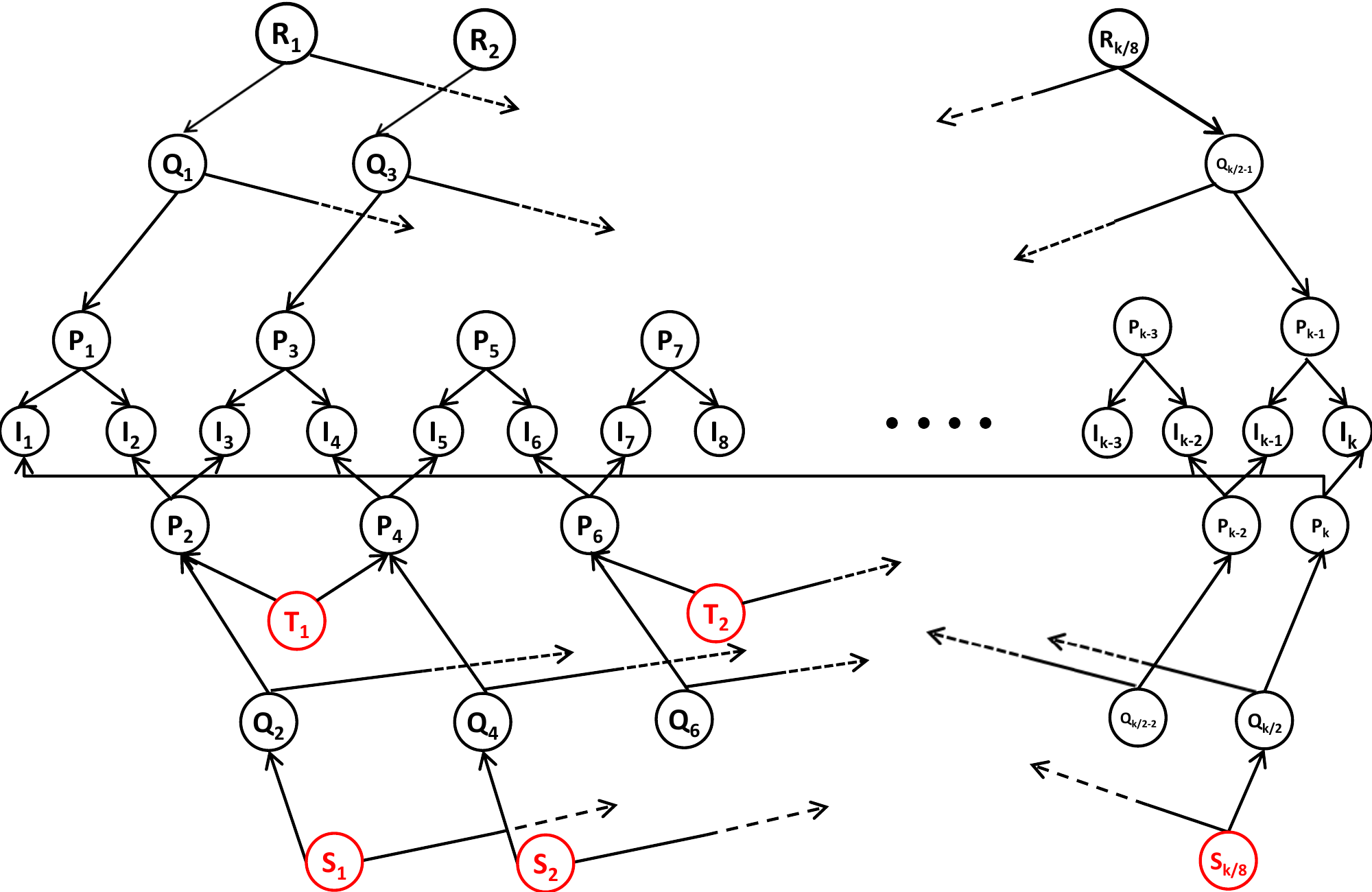}
		\caption{The Six Erasure Code : Note that the new parities $T$ covers only half of $P$ parities in the bottom (i.e. $P_2,P_4 \dots P_\frac{k}{2}$). Hence they are $k/8$ in number.}
		\label{fig:six_erasure}
	\end{figure}
	\begin{align*}
		\text{Rate of the code } &= \frac{k}{k+k+(k/2)+(k/8)+(k/8)+(k/8)}\\
		&=0.3478
	\end{align*}
	\subsection{Construction for $t=7,r=2$} \label{sec:t7r2}
A $(2,7)_{seq}$ code can be constructed from the $(2,6)_{seq}$ code constructed as described above, by adding additional parities $U_1 \dots U_{(k/16)}$ and $V_1 \dots V_{(k/16)}$. $U_i$ is the parity of $T_i$ and $T_{i+\frac{k}{16}}$. $V_i$ is the parity of $S_i$ and $S_{i+\frac{k}{16}}$ Here we need the additional requirement of $16|k$.
	\begin{figure}[ht!]
		\centering
		\includegraphics[height=2.5in]{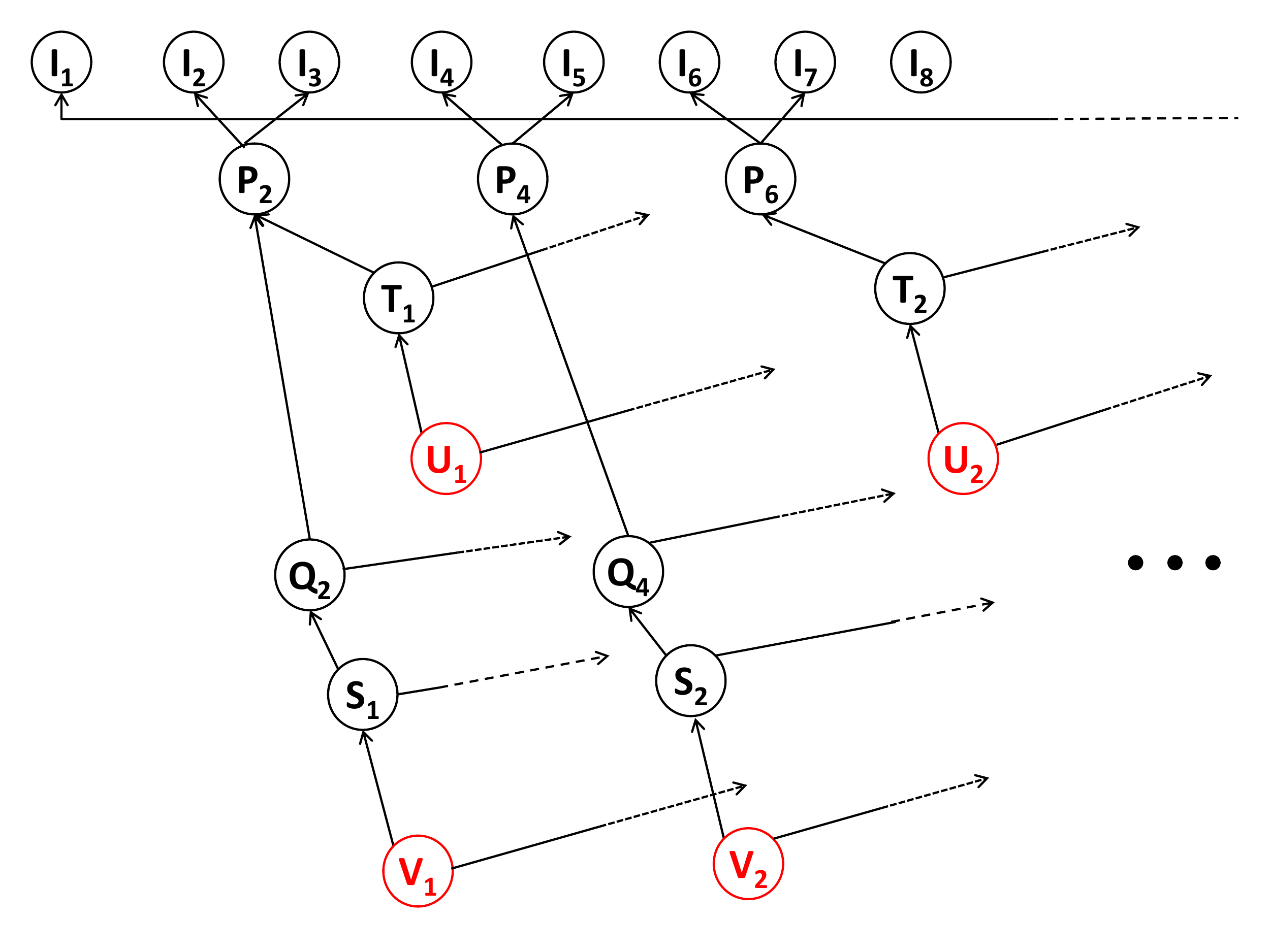}
		\caption{The seven Erasure Code (partial diagram)}
		\label{fig:seven_erasure}
	\end{figure}
	\\
	
	\text{Rate of the code is : } \\ $\frac{k}{k+k+(k/2)+(k/8)+(k/8)+(k/8)+(k/16)+(k/16)}
	=0.3333$.\\
	
   The above rate is same as the rate $\frac{r-1}{r+1}=0.3333$ achieved by construction in \cite{RawMazVis} based on bipartite graphs (Section VI-A of \cite{RawMazVis}) but our construction requires a minimum block length of 48 (for $k=16$) whereas the construction mentioned in  \cite{RawMazVis} requires block length of $(r+1)^{t-4-\lfloor\frac{t-2}{4} \rfloor+2}=3^4=81$ \cite{LazFel}.  


\begin{table}[] 
	\begin{centering}
	\begin{tabular}{||c|c|c|c|c|| }
		\hline
		$t$ & 4 & 5 & 6 & 7\\
		\hline
		\hline
		Bound \eqref{tamo_rate} & 0.4063 & 0.3694 & 0.3410 & 0.3183 \\
		\hline
		Rate of codes in Section \ref{sec:r_eq_2} & 0.4000 & 0.3810 & 0.3478 & 0.3333  \\
		\hline
	\end{tabular}
	\caption{Comparison of code rate of constructions for $r=2$ versus bound in \cite{TamBarFro}.}
	\label{table:high_rate}
	\end{centering}
\end{table}

	\section{Constructions for General $t$}\label{sec:general_t}
    
	\subsection{Construction using Orthogonal Latin Squares}
	Let $\{L_1 \cdots L_{t-2}\}$ be a set of $(t-2)$ pairwise orthogonal Latin squares of size $(r \times r)$.  Necessarily, $(t-2) \leq (r-1)$.  Let $L_{t-1},L_{t}$ be an additional two $(r \times r)$ matrices with $(i,j)^{th}$ entries given by $L_{t-1,(i,j)}=i$, \ $L_{t,(i,j)}=j$.   While $L_{t-1},L_t$ are not Latin squares, any two squares in the enlarged set ${\cal L} = \{L_1, \cdots, L_{t+1}\}$ continue to be pairwise orthogonal.   
	%
	%
	Next, 	let $A$ be the $(rt \times r^2)$ matrix constructed from ${\cal L}$ as follows.  The columns of $A$ are indexed by a pair $(a,b)$ of coordinates, $1 \leq a,b \leq r$.  Then
	\bean
	A_{i,(a,b)} & = & \left\{ \begin{array}{rl}    1 & L_{\lceil \frac{i}{r} \rceil,(a,b)} = i \hspace*{-0.1in} \pmod{r}+1 \\ 0 & \text{else} . \end{array} \right. 
	\eean
	Let $H$ be the parity-check matrix given by
	\bea
	H & = & \scalemath{0.8}{ \left[ 
		\begin{array}{c | c c c c | c}
			& I_r  & 0 & \dots & 0 & \\
			A \hspace{10pt}	& 0 &  \ddots& \vdots & 0 & 0\\
			& 0 &  \dots & 0 & I_r & \\
			\hline
			0 & \underline{1} & 0 & \dots  & 0 & 1 \\
		\end{array}
		\right] }.\label{eq;orth_lat} 
	\eea 
	where $\underline{1}$ denotes the $(1 \times r)$ vector of all ones, $I_r$ denotes the identity matrix of order $r$, repeated $t$ times along the diagonal as shown.
    \begin{claim}
    	The code with parity check matrix as defined in \eqref{eq;orth_lat} is a $(r,t+1)_{seq}$ code, when $t$ is even. 
    \end{claim}
    \begin{proof}
    	Throughout this proof for the matrix H, we identify the Support(R) for some row R of H with indices of columns in which R has non zero entries and the code symbols corresponding to the indices of columns in which R has non zero entries synonymously.
    	We can divide the rows of $A$ into $t$ sets, $S_i,1 \leq i \leq t$, $S_i = \{R_{(i-1)r+1},\cdots ,R_{ir}\}$ (where $R_i$ indicates the $i^{th}$ row of $A$) and the support sets of the rows in $S_i$ i.e., $Support(R_{(i-1)r+j}),1 \leq j \leq r$  depends on $L_i$ and form a partition of $r^2$ columns of $A$. Each row of $A$ contain exactly $r$, non zero entries. Clearly, the support of $2$ distinct rows $R_i,R_j \in S_k$ for some $k$, are disjoint. Since the elements in ${\cal L}$ are mutually orthogonal, the support of two rows $R_i,R_j, i \neq j$ such that $R_i \in S_{k_1},R_j \in S_{k_2}, k_1 \neq k_2$ will have $|Support(R_i) \cap Support(R_j)|\leq 1$. These two observations prove that each symbol corresponding to the columns of $A$ is protected by $t$ orthogonal parities in $A$ of weight $r$ each. Hence, $A$ denote the parity check matrix of a $(r-1,t)_{par}$ code. This also implies that each symbol corresponding to the columns of $A$ is protected by $t$ orthogonal parities in $H$ of weight $r+1$.
    	
    	Let the set of first $r^2$ symbols of the code (corresponding to the columns of $A$) be denoted by $B_1$ and the set of remaining $rt+1$ code symbols be denoted by $B_2$.
    	Assume that there are $s$ erased symbols $T_1=\{x_1,...,x_s\}$ in $B_1$ and $t-s+1$ erased symbols $T_2=\{x_{s+1},...,x_{t+1}\}$ in $B_2$.
    	
    	Consider the case when $s=1$. Let $T_1=\{x_1\}$. $x_1$ is covered by $t$ orthogonal parities corresponding to rows $R'_{a_i}, 1 \leq i \leq t$ where $1 \leq a_i \leq rt$ ($R'_j$ denotes the $j^{th}$ row of $H$) with supports $D_{a_i}, 1\leq i \leq t$. If $D_{a_i} \cap T_2 = \emptyset$ for some i,  then $x_1$ can be recovered. Hence, assume that wlog $D_{a_i} \cap T_2=\{x_{1+i}\}, \forall 1 \leq i \leq t$. Since $R_{a_i} \in S_1$ for some i and  $D_{a_i} \cap T_2=\{x_{1+i}\}$ for that i, $x_{1+i}$ can be recovered using the last row parity of $H$. This recovered symbol can be used to recover $x_1$, and subsequently all symbols can be recovered.
    	
    	Now consider the case when $1<s<t+1$. Let $x_1 \in T_1$ be protected by $t$ orthogonal parities corresponding to rows $R'_{a_i}, 1 \leq i \leq t$ where $1 \leq a_i \leq rt$ with supports $D_{a_i}, 1\leq i \leq t$. If $D_{a_i}  \cap (T_1-\{x_1\} \cup T_2) = \emptyset$ for some i,  then $x_1$ can be recovered. Hence, assume that wlog $D_{a_i} \cap (T_1-\{x_1\} \cup T_2)=\{x_{1+i}\},  \forall 1 \leq i \leq t$. Since $s>1$, $x_2 \in T_1$. $x_2$ is also protected by $t$ orthogonal parities corresponding to rows $R'_{b_i}, 1 \leq i \leq t$ where $1 \leq b_i \leq rt$ with supports $D_{b_i}, 1\leq i \leq t$. If $D_{b_i} \cap (T_1-\{x_2\} \cup T_2) = \emptyset$ for some i,  then $x_2$ can be recovered. Hence, assume that wlog $|D_{b_i} \cap (T_1-\{x_2\})|=1$, $|D_{b_i} \cap T_2|=0$,  $\forall 1 \leq i \leq s-1$ and $|D_{b_i} \cap T_2|=1$,$|D_{b_i} \cap (T_1-\{x_2\})|=0$,   $\forall s \leq i \leq t$.   Hence $x_{s+j_1} \in D_{b_{j_2}}$ for some $s \leq j_2 \leq t$ and $1 \leq j_1 \leq t-s+1$. Since $x_{s+j_1} \in Support(R'_k)$ for exactly one $k$ in $1 \leq k \leq rt$ and $x_{s+j_1} \in D_{a_{s+j_1-1}}$,$x_{s+j_1} \in D_{b_{j_2}}$, this implies $D_{a_{s+j_1-1}} = D_{b_{j_2}}$ as $x_{s+j_1} \in B_2$. But this is not possible as $x_1 \in D_{a_{s+j_1-1}}$ and $x_1 \notin D_{b_{j_2}}$ as $s \leq j_2 \leq t$. Hence the symbol $x_2$ can be recovered. Similarly all erased symbols can be recovered.
    	
    	Now, consider the case when $s=t+1$. Let the $t$ orthogonal parities protecting $x_i$ be corresponding to $R^i_{j}$ (rows of $H$), $\forall 1 \leq i \leq t+1, 1 \leq j \leq t$. If one of the symbol can be recovered then the rest of the symbols can be recovered by their $t$ orthogonal parities. Hence for none of the symbols to be recoverable, the only possibility is that $|Support(R^i_{j}) \cap T_1-\{x_i\}|=1$,  $\forall 1 \leq i \leq t+1, 1 \leq j \leq t$. Hence $|Support(R) \cap T_1| \in \{0,2\}$ for any row $R$ of $H$. If we consider rows $R'_1,..,R'_r$ (first $r$ rows) of $H$ with supports $D'_1,..,D'_r$, $D'_i \cap D'_j = \emptyset$ and $[r^2] \subset \cup^{r}_{i=1} D'_i$. Now since $\sum_i |D'_i \cap T_1| =t+1$ which is odd but $\sum_i |D'_i \cap T_1|$ is even as $|D'_i \cap T_1| \in \{0,2\}$, this leads to a contradiction. Hence one of the erased symbols can be recovered and subsequently the remaining erased symbols can be recovered.
    \end{proof}

	This code has block length $r^2+rt+1$ and rate $\frac{r^2}{r^2+rt+1}$.  In comparison, the rate of the construction given in \cite{WanZhaLiu_Arb_Locality} is $\frac{r}{r+t+1}$ with a block length of $\binom{r+t+1}{t+1}$.  Thus our construction achieves both a better rate and a smaller block length , made possible by adopting a sequential approach to recovery as opposed to using orthogonal parities. The construction in \cite{RawMazVis} based on bipartite graphs (Section VI-A of \cite{RawMazVis}) has a higher rate of $\frac{r-1}{r+1}$ but requires block length of size $(r+1)^{t-4-\lfloor\frac{t-2}{4} \rfloor+2}$ \cite{LazFel}.  

	\subsection{Construction using product of sequential codes}\label{sec:prod_code}
Let $\mathcal{C}_1$ be a $[n_1,k_1]$ $(r,t_1)_{seq}$ code and $\mathcal{C}_2$, be a $[n_2,k_2]$ $(r,t_2)_{seq}$ code. Let $\mathcal{C}$ be the code obtained by taking the product of $\mathcal{C}_1$ and $\mathcal{C}_2$ which will be a $[n_1 n_2,k_1 k_2]$ code with locality $r$. The sequential erasure correcting capability of $\mathcal{C}$ is given by:
\begin{claim}
The code $\mathcal{C}$ constructed as described above can recover from $(t_1+1)(t_2+1)-1$ erasures sequentially. 	
\end{claim} 
		\begin{figure}[h]
			\centering
			\includegraphics[height=1.5in]{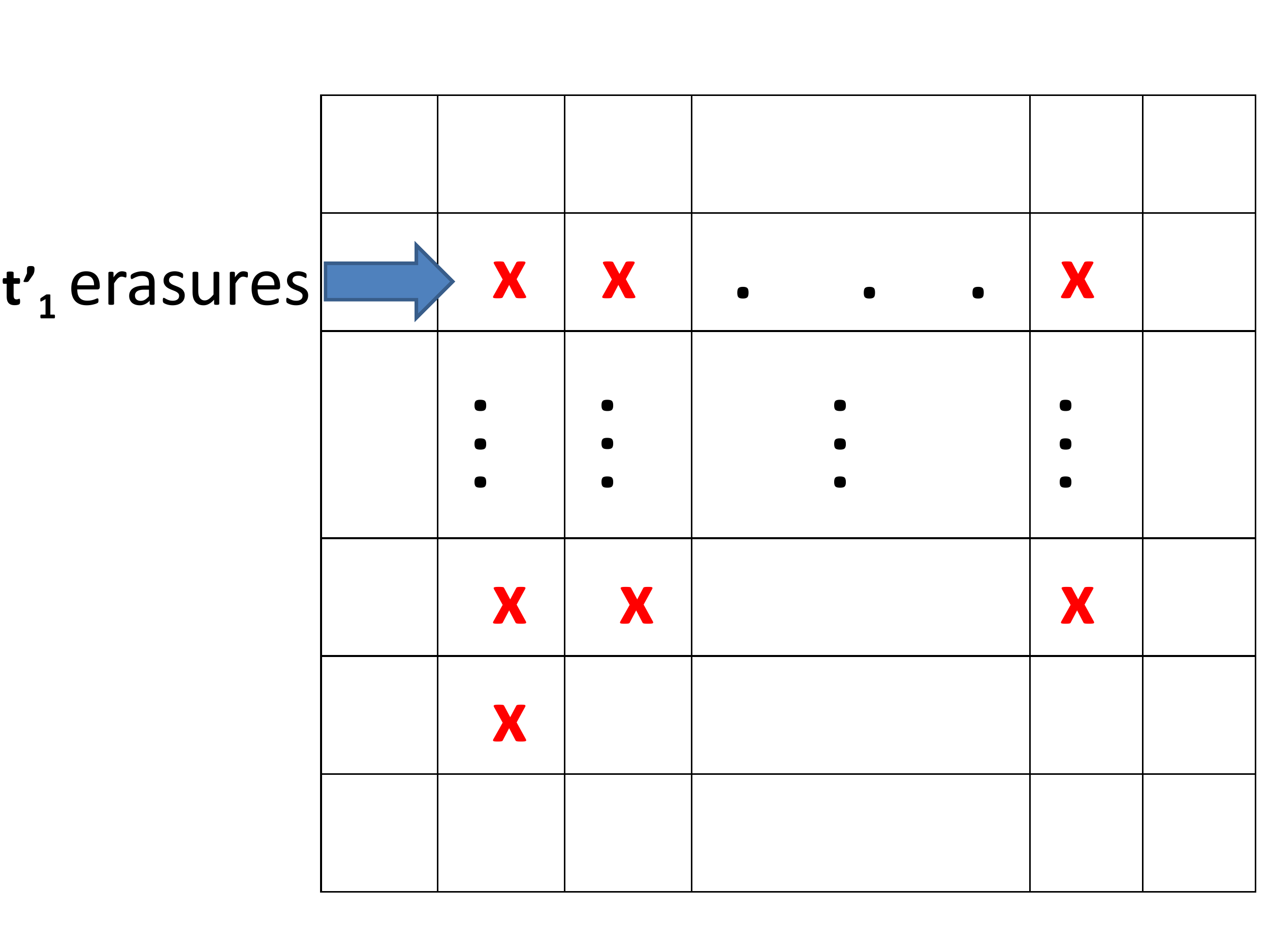}
			\caption{Product code : 'X' indicates node failures.}
			\label{fig:prod_error}
		\end{figure}

	\begin{proof}
		Consider a codeword in $\mathcal{C}$ as shown in Figure \ref{fig:prod_error} consisting of a $(r,t_1)_{seq}$ code in the rows and $(r,t_2)_{seq}$ code in the columns. Assume that there are $t_1'$ erasures in a row of the codeword. If $t_1'\leq t_1$, then the nodes can be recovered using $\mathcal{C}_1$. Hence assume that $t_1'>t_1$. $\mathcal{C}_2$ can recover one of these symbols if the number of erasures in the column corresponding to that symbol is less than or equal to $t_2$. Hence, the code $\mathcal{C}$ fails to repair the erased symbols only if the number of  erasures is greater than or equal to $(t_1+1)(t_2+1)$, which completes the proof.
	\end{proof}
	
	\subsubsection{Suboptimality of the product of $\delta$ number of [3,2] single parity check codes for $\delta \geq 9$ (\cite{SonYue_Binary_local_repair})} 
           Take the product of the code given in Section \ref{sec:t7r2} for $k=16$  with the  product of three [7,3] Simplex codes i.e.,

           Let $\cal{C}$ = Product of (Code given in Section \ref{sec:t7r2} for $k=16$, [7,3] Simplex code, [7,3] Simplex code, [7,3] Simplex code)

          The rate of the resulting code $\cal{C}$ will be : 0.026239067 with locality $r=2$ and $t=2^9 -1$ and block length $n = 16464$.

If we take the product of nine $[3,2]$ single parity check codes, the resulting code will have $n=3^9 = 19683$, rate = 0.0260122 and $r=2$, $t=2^9-1$.

Thus $\cal{C}$ achieves a better rate with a smaller block length than the code given by product of nine [3,2] single parity check codes. Beyond this for $t=2^\delta-1$ for $\delta > 9$, we can simply take the code given by product of $\cal{C}$ and the code obtained from product of $\delta-9$ number of $[3,2]$ single parity check codes, and achieve a better rate with a smaller block length than the the product of $\delta$ number of $[3,2]$ single parity check codes with the same locality OF $r=2$ and the same erasure correcting capability of $t=2^\delta-1$.

Hence product code is not optimal for sequential recovery, for $r=2$,$t=2^\delta-1$ and $\delta \geq 9$.\\

	\subsubsection{An Example Construction}
	
	Let $H_t$ denote the parity check matrix of a $(r,t)$-sequential erasure correcting code with parameters $[n',k',d']$. Construct a new matrix $H$ as shown below. 
	\bea
	H  = 
	\begin{pmatrix}
		H_t & 0 & \cdots & 0 & 0 \\
		0 & H_t & \cdots & 0 & 0 \\
		\vdots & \vdots & \vdots & \vdots & \vdots \\
		0 & 0 & \cdots & H_t & 0\\
		I_{n'} & I_{n'} & \cdots & I_{n'} & I_{n'}
	\end{pmatrix}, \label{eq:prod_mx}
	\eea 
	
	$H_t$ is repeated $r$ times along the diagonal. $I_{n'}$ denote the $n'\times n'$identity matrix.
	\begin{claim}
		A code with parity check matrix $H$ as defined by \eqref{eq:prod_mx} has all symbol locality $r$ and can correct $2t+1$ erasures using sequential approach.
	\end{claim}
	
	\begin{proof}
		Proof follows by observing that the given parity check matrix is the parity check matrix of the code obtained by taking the product of single parity check code and the $[n',k',d']$ code.
	\end{proof}
	The resulting code will have the following parameters 
	\begin{align*}
		\text{Block Length }&= (r+1)n'\\
		\text{Dimension }&= rk'\\
		\text{Minimum Distance } &=2t+2
	\end{align*}

\end{document}